\documentclass[peerreview,a4paper,12pt]{IEEEtran}

\usepackage{amsthm}
\usepackage{amsmath}
\usepackage{dsfont}
\usepackage{amsfonts,amssymb,amsmath}
\usepackage{graphicx}
\usepackage{epsfig}
\usepackage[applemac]{inputenc}
\usepackage{latexsym}

\newtheorem{mydef}{Definition}
\newtheorem{mycor}{Corollary}
\newtheorem{myprop}{Proposition}
\newtheorem{mythe}{Theorem}
\newtheorem{mylem}{Lemma}
\newtheorem{mynot}{Notation}
\newtheorem{myrem}{Remark}

\newtheorem{myconj}{Conjecture}

\DeclareMathOperator*{\DMC}{DMC}
\DeclareMathOperator*{\Proj}{Proj}

\DeclareMathOperator*{\srank}{srank}

\DeclareMathOperator*{\opt}{opt}
\DeclareMathOperator*{\ach}{ach}
\DeclareMathOperator*{\CPC}{CPC}
\DeclareMathOperator*{\PC}{PC}

\begin{document}

\sloppy

%% Paper Title
%% You can use linebreaks \\ within to get better formatting as
%% desired. 
\title{A Characterization of the Shannon Ordering of Communication Channels}

%% Author names and affiliations:
%%
%% Avoiding spaces at the end of the author lines is not a problem with
%% conference papers because we don't use \thanks or \IEEEmembership.
%%
%% For several authors with only one affiliation:
%%
% \author{
%   \IEEEauthorblockN{Hui-Ting Chang and Stefan M.~Moser}
%   \IEEEauthorblockA{Department of Electrical and Computer Engineering\\
%     National Chiao Tung University (NCTU)\\
%     Hsinchu, Taiwan\\
%     Email: \{email-of-hui-ting,email-of-stefan\}@ieee.org} 
% }
%%
%% For up to three affiliations:
%%
\author{
  \IEEEauthorblockN{Rajai Nasser\\}
  \IEEEauthorblockA{
EPFL, Lausanne, Switzerland\\
Email: rajai.nasser@epfl.ch} 
}
%%
%% For over three affiliations, or if they all won't fit within the width
%% of the page, use this alternative format:
%%
% \author{
%   \IEEEauthorblockN{
%     Michael Shell\IEEEauthorrefmark{1},
%     Homer Simpson\IEEEauthorrefmark{2},
%     James Kirk\IEEEauthorrefmark{3}, 
%     Montgomery Scott\IEEEauthorrefmark{3} and
%     Eldon Tyrell\IEEEauthorrefmark{4}}
%   \IEEEauthorblockA{
%     \IEEEauthorrefmark{1}School of Electrical and Computer Engineering\\
%     Georgia Institute of Technology, Atlanta, Georgia 30332--0250\\ 
%     Email: see http://www.michaelshell.org/contact.html}
%   \IEEEauthorblockA{
%     \IEEEauthorrefmark{2}Twentieth Century Fox, Springfield, USA\\
%     Email: homer@thesimpsons.com}
%   \IEEEauthorblockA{
%     \IEEEauthorrefmark{3}Starfleet Academy, San Francisco, California 96678-2391\\
%     Telephone: (800) 555--1212, Fax: (888) 555--1212}
%   \IEEEauthorblockA{
%     \IEEEauthorrefmark{4}Tyrell Inc., 123 Replicant Street, Los Angeles, California 90210--4321}
% }

%% Use for special paper notices

%% To balance the two columns, you should reduce the text-height of
%% the last page using the following command:
%%%%%%%%%%%%%%%%%%%%%%%%%%%%%%%%%%%%%%%%%%%%%%%%%%%%%%%%%%%%%%%%%%%%%
%\addtolength{\textheight}{-9.35cm}
%%%%%%%%%%%%%%%%%%%%%%%%%%%%%%%%%%%%%%%%%%%%%%%%%%%%%%%%%%%%%%%%%%%%%
%% with an appropriate value. This command must be place on the second
%% last page, i.e., for a one-page abstract here, for a two-page
%% abstract right after the \maketitle command.

%% Create the title:
\maketitle

%% Abstract: 
%% For the final version of the accepted paper, please make sure you
%% remove the comment "THIS PAPER IS ELIGIBLE FOR THE STUDENT PAPER
%% AWARD."
%%
\begin{abstract}
The ordering of communication channels was first introduced by Shannon. In this paper, we aim to find a characterization of the Shannon ordering. We show that $W'$ contains $W$ if and only if $W$ is the skew-composition of $W'$ with a convex-product channel. This fact is used to derive a characterization of the Shannon ordering that is similar to the Blackwell-Sherman-Stein theorem. Two channels are said to be Shannon-equivalent if each one is contained in the other. We investigate the topologies that can be constructed on the space of Shannon-equivalent channels. We introduce the strong topology and the BRM metric on this space. Finally, we study the continuity of a few channel parameters and operations under the strong topology.
\end{abstract}

\section{Introduction}

The ordering of communication channels was first introduced by Shannon in \cite{ShannonDegrad}. A channel $W'$ is said to contain another channel $W$ if $W$ can be simulated from $W'$ by randomization at the input and the output using a shared randomness between the transmitter and the receiver. Shannon showed that the existence of an $(n,M,\epsilon)$ code for $W$ implies the existence of an $(n,M,\epsilon)$ code for $W'$.

Another ordering that has been well studied is the degradedness between channels. A channel $W$ is said to be degraded from another channel $W'$ if $W$ can be simulated from $W'$ by randomization at the output, or more precisely, if $W$ can be obtained from $W'$ by composing it with another channel. It is easy to see that degradedness is a special case of Shannon's ordering. One can trace the roots of the notion of degradedness to the seminal work of Blackwell in the 1950's about comparing statistical experiments \cite{blackwell1951}. Note that in the Shannon ordering, the input and output alphabets need not be the same, whereas in the degradedness definition, we have to assume that $W$ and $W'$ share the same input alphabet $\mathcal{X}$ but they can have different output alphabets. A characterization of degradedness is given by the famous Blackwell-Sherman-Stein (BSS) theorem \cite{blackwell1951}, \cite{Sherman}, \cite{Stein}.

In \cite{RajInputDegrad}, we introduced the input-degradedness ordering of communication channels. A channel $W$ is said to be input-degraded from another channel $W'$ if $W$ can be simulated from $W'$ by randomization at the input. Note that $W$ and $W'$ must have the same output alphabet, but they can have different input alphabets. In \cite{RajInputDegrad}, we provided two characterizations of input-degradedness, one of which is similar to the BSS theorem. The main purpose of this paper is to find a characterization of the Shannon ordering that is similar to the BSS theorem.

In \cite{RaginskyShannon}, Raginsky introduced the Shannon deficiency which compares a particular channel with the Shannon-equivalence class of another channel. The Shannon deficiency is not a metric that compares two Shannon-equivalence classes of channels.

In \cite{RajDMCTop} and \cite{RajContTop}, we constructed topologies for the space of equivalent channels and studied the continuity of various channel parameters and operations under these topologies. In this paper, we show that some of the results in \cite{RajDMCTop} and \cite{RajContTop} can be replicated (with some variation) for the space of Shannon-equivalent channels.

\section{Preliminaries}

We assume that the reader is familiar with the basic concepts of general topology. The main concepts and theorems that we need can be found in the preliminaries section of \cite{RajDMCTop}.

\subsection{Set-theoretic notations}

For every integer $n>0$, we denote the set $\{1,\ldots,n\}$ as $[n]$.

Let $(A_i)_{i\in I}$ be a collection of arbitrary sets indexed by $I$. The \emph{disjoint union} of $(A_i)_{i\in I}$ is defined as $\displaystyle \coprod_{i\in I} A_i=\bigcup_{i\in I}(A_i\times\{i\})$. For every $i\in I$, the $i^{th}$-\emph{canonical injection} is the mapping $\phi_i:A_i\rightarrow \displaystyle\coprod_{j\in I} A_j$ defined as $\phi_i(x_i)=(x_i,i)$. If no confusions can arise, we can identify $A_i$ with $A_i\times\{i\}$ through the canonical injection. Therefore, we can see $A_i$ as a subset of $\displaystyle\coprod_{j\in I} A_j$ for every $i\in I$.

Let $R$ be an equivalence relation on $T$. For every $x\in T$, the set $\hat{x}=\{y\in T:\; x R y\}$ is the \emph{$R$-equivalence class} of $x$. The collection of $R$-equivalence classes, which we denote as $T/R$, forms a partition of $T$, and it is called the \emph{quotient space of $T$ by $R$}. The mapping $\Proj_R:T\rightarrow T/R$ defined as $\Proj_R(x)=\hat{x}$ for every $x\in T$ is the \emph{projection mapping onto $T/R$}.

\subsection{Measure theoretic notations}

The set of probability measures on a measurable space $(M,\Sigma)$ is denoted as $\mathcal{P}(M,\Sigma)$. For every $P_1,P_2\in\mathcal{P}(M,\Sigma)$, the \emph{total variation distance} between $P_1$ and $P_2$ is defined as:
$$\|P_1-P_2\|_{TV}=\sup_{A\in\Sigma}|P_1(A)-P_2(A)|.$$

If $\mathcal{X}$ is a finite set, we denote the set of probability distributions on $\mathcal{X}$ as $\Delta_{\mathcal{X}}$. We always endow $\Delta_{\mathcal{X}}$ with the total variation distance and its induced topology.

\subsection{Quotient topology}

Let $(T,\mathcal{U})$ be a topological space and let $R$ be an equivalence relation on $T$. The \emph{quotient topology} on $T/R$ is the finest topology that makes the projection mapping $\Proj_R$ onto the equivalence classes continuous. It is given by
$$\mathcal{U}/R=\left\{\hat{U}\subset T/R:\;\textstyle\Proj_R^{-1}(\hat{U})\in \mathcal{U}\right\}.$$

\begin{mylem}
\label{lemQuotientFunction}
Let $f:T\rightarrow S$ be a continuous mapping from $(T,\mathcal{U})$ to $(S,\mathcal{V})$. If $f(x)=f(x')$ for every $x,x'\in T$ satisfying $x R x'$, then we can define a \emph{transcendent mapping} $f:T/R\rightarrow S$ such that $f(\hat{x})=f(x')$ for any $x'\in\hat{x}$. $f$ is well defined on $T/R$ . Moreover, $f$ is a continuous mapping from $(T/R,\mathcal{U}/R)$ to $(S,\mathcal{V})$.
\end{mylem}

Let $(T,\mathcal{U})$ and $(S,\mathcal{V})$ be two topological spaces and let $R$ be an equivalence relation on $T$. Consider the equivalence relation $R'$ on $T\times S$ defined as $(x_1,y_1) R' (x_2,y_2)$ if and only if $x_1 R x_2$ and $y_1=y_2$. A natural question to ask is whether the canonical bijection between $\big((T/R)\times S,(\mathcal{U}/R)\otimes \mathcal{V} \big)$ and $\big((T\times S)/R',(\mathcal{U}\otimes \mathcal{V})/R' \big)$ is a homeomorphism. It turns out that this is not the case in general. The following theorem, which is widely used in algebraic topology, provides a sufficient condition:

\begin{mythe}
\label{theQuotientProd}
\cite{Engelking}
If $(S,\mathcal{V})$ is locally compact and Hausdorff, then the canonical bijection between $\big((T/R)\times S,(\mathcal{U}/R)\otimes \mathcal{V} \big)$ and $\big((T\times S)/R',(\mathcal{U}\otimes \mathcal{V})/R' \big)$ is a homeomorphism.
\end{mythe}

\begin{mycor}
\label{corQuotientProd}
\cite{RajContTop} Let $(T,\mathcal{U})$ and $(S,\mathcal{V})$ be two topological spaces, and let $R_T$ and $R_S$ be two equivalence relations on $T$ and $S$ respectively. Define the equivalence relation $R$ on $T\times S$ as $(x_1,y_1) R (x_2,y_2)$ if and only if $x_1 R_T x_2$ and $y_1R_S y_2$. If $(S,\mathcal{V})$ and $(T/R_T,\mathcal{U}/R_T)$ are locally compact and Hausdorff, then the canonical bijection between $\big((T/R_T)\times (S/R_S),(\mathcal{U}/R_T)\otimes (\mathcal{V}/R_S) \big)$ and $\big((T\times S)/R,(\mathcal{U}\otimes \mathcal{V})/R \big)$ is a homeomorphism.
\end{mycor}

\subsection{The space of channels from $\mathcal{X}$ to $\mathcal{Y}$}

Let $\DMC_{\mathcal{X},\mathcal{Y}}$ be the set of all channels having $\mathcal{X}$ as input alphabet and $\mathcal{Y}$ as output alphabet. For every $W,W'\in\DMC_{\mathcal{X},\mathcal{Y}}$, define the distance between $W$ and $W'$ as:
$$d_{\mathcal{X},\mathcal{Y}}(W,W')=\frac{1}{2} \max_{x\in\mathcal{X}}\sum_{y\in\mathcal{Y}}|W'(y|x)-W(y|x)|.$$

Throughout this paper, we always associate the space $\DMC_{\mathcal{X},\mathcal{Y}}$ with the metric distance $d_{\mathcal{X},\mathcal{Y}}$ and the metric topology $\mathcal{T}_{\mathcal{X},\mathcal{Y}}$ induced by it. It is easy to see that $\mathcal{T}_{\mathcal{X},\mathcal{Y}}$ is the same as the topology inherited from the Euclidean topology of $\mathbb{R}^{\mathcal{X}\times\mathcal{Y}}$ by relativization. It is also easy to see that the metric space $\DMC_{\mathcal{X},\mathcal{Y}}$ is compact and path-connected (see \cite{RajDMCTop}).

For every $W\in\DMC_{\mathcal{X},\mathcal{Y}}$ and every $V\in\DMC_{\mathcal{Y},\mathcal{Z}}$, define the composition $V\circ W\in\DMC_{\mathcal{X},\mathcal{Z}}$ as
$$(V\circ W)(z|x)=\sum_{y\in\mathcal{Y}}V(z|y)W(y|x),\;\;\forall x\in\mathcal{X},\;\forall z\in\mathcal{Z}.$$

For every mapping $f:\mathcal{X}\rightarrow\mathcal{Y}$, define the deterministic channel $D_f\in\DMC_{\mathcal{X},\mathcal{Y}}$ as $$D_f(y|x)=\begin{cases}1\quad&\text{if}\;y=f(x),\\0\quad&\text{otherwise}.\end{cases}$$
It is easy to see that if $f:\mathcal{X}\rightarrow \mathcal{Y}$ and $g:\mathcal{Y}\rightarrow \mathcal{Z}$, then $D_g\circ D_f=D_{g\circ f}$.

\subsection{Channel parameters}

The \emph{capacity} of a channel $W\in\DMC_{\mathcal{X},\mathcal{Y}}$ is denoted as $C(W)$.

An \emph{$(n,M)$-encoder} on the alphabet $\mathcal{X}$  is a mapping $\mathcal{E}:\mathcal{M}\rightarrow\mathcal{X}^n$ such that $|\mathcal{M}|=M$. The set $\mathcal{M}$ is the \emph{message set} of $\mathcal{E}$, $n$ is the \emph{blocklength} of $\mathcal{E}$, $M$ is the \emph{size} of $\mathcal{E}$, and $\frac{1}{n}\log M$ is the \emph{rate} of $\mathcal{E}$ (measured in nats). The \emph{error probability of the ML decoder for the encoder $\mathcal{E}$ when it is used for a channel $W\in\DMC_{\mathcal{X},\mathcal{Y}}$} is given by:
$$P_{e,\mathcal{E}}(W)=1-\frac{1}{M}\sum_{y_1^n\in\mathcal{Y}^n} \max_{m\in\mathcal{M}}\left\{\prod_{i=1}^n W(y_i|\mathcal{E}_i(m))\right\},$$
where $(\mathcal{E}_1(m),\ldots,\mathcal{E}_n(m))=\mathcal{E}(m)$.

The \emph{optimal error probability of $(n,M)$-encoders for a channel $W$} is given by:
$$P_{e,n,M}(W)=\min_{\substack{\mathcal{E}\;\text{is an}\\(n,M)\text{-encoder}}}P_{e,\mathcal{E}}(W).$$

\subsection{Channel operations}

For every $W_1\in \DMC_{\mathcal{X}_1,\mathcal{Y}_1}$ and $W_2\in \DMC_{\mathcal{X}_2,\mathcal{Y}_2}$, define the \emph{channel sum} $W_1 \oplus W_2\in \DMC_{\mathcal{X}_1\coprod\mathcal{X}_2,\mathcal{Y}_1\coprod\mathcal{Y}_2}$ of $W_1$ and $W_2$ as:
$$(W_1\oplus W_2)(y,i|x,j)=\begin{cases}W_i(y|x)\quad&\text{if}\;i=j,\\
0&\text{otherwise},\end{cases}$$
where $\mathcal{X}_1\coprod\mathcal{X}_2=(\mathcal{X}_1\times\{1\})\cup(\mathcal{X}_2\times\{2\})$ is the disjoint union of $\mathcal{X}_1$ and $\mathcal{X}_2$. $W_1\oplus W_2$ arises when the transmitter has two channels $W_1$ and $W_2$ at his disposal and he can use exactly one of them at each channel use.

We define the \emph{channel product} $W_1\otimes W_2\in \DMC_{\mathcal{X}_1\times\mathcal{X}_2,\mathcal{Y}_1\times\mathcal{Y}_2}$ of $W_1$ and $W_2$ as:
$$(W_1\otimes W_2)(y_1,y_2|x_1,x_2)=W_1(y_1|x_1)W_2(y_2|x_2).$$
$W_1\otimes W_2$ arises when the transmitter has two channels $W_1$ and $W_2$ at his disposal and he uses both of them at each channel use. Channel sums and products were first introduced by Shannon in \cite{ChannelSumProduct}.

\section{Shannon ordering and Shannon-equivalence}

\label{secShannonOrd}

Let $\mathcal{X},\mathcal{X}',\mathcal{Y}$ and $\mathcal{Y}'$ be three finite sets. Let $W\in \DMC_{\mathcal{X},\mathcal{Y}}$ and $W'\in \DMC_{\mathcal{X}',\mathcal{Y}'}$. We say that $W'$ contains $W$ if there exist $n$ pairs of channels $(R_i,T_i)_{1\leq i\leq n}$ and a probability distribution $\alpha\in\Delta_{[n]}$ such that $R_i\in\DMC_{\mathcal{X},\mathcal{X}'}$ and $T_i\in\DMC_{\mathcal{Y}',\mathcal{Y}}$ for every $1\leq i\leq n$, and $\displaystyle W=\sum_{i=1}^n \alpha(i)T_i\circ W'\circ R_i$, i.e.,
$$W(y|x)=\sum_{i=1}^n\alpha(i)\sum_{\substack{x'\in\mathcal{X}',\\y'\in\mathcal{Y}'}}T_i(y|y')W'(y'|x')R_i(x'|x).$$
The channels $W$ and $W'$ are said to be \emph{Shannon-equivalent} if each one contains the other.

A channel $V\in\DMC_{\mathcal{X}\times\mathcal{Y}',\mathcal{X}'\times\mathcal{Y}}$ is said to be a \emph{convex-product channel} if it is the convex combination of the products of channels in $\DMC_{\mathcal{X},\mathcal{X}'}$ with channels in $\DMC_{\mathcal{Y}',\mathcal{Y}}$. More precisely, $V\in\DMC_{\mathcal{X}\times\mathcal{Y}',\mathcal{X}'\times\mathcal{Y}}$ is a convex-product channel if there exist $n$ pairs of channels $(R_i,T_i)_{1\leq i\leq n}$ and a probability distribution $\alpha\in\Delta_{[n]}$ such that $R_i\in\DMC_{\mathcal{X},\mathcal{X}'}$ and $T_i\in\DMC_{\mathcal{Y}',\mathcal{Y}}$ for every $1\leq i\leq n$, and
$$V(x',y|x,y')=\sum_{i=1}^n\alpha(i)R_i(x'|x)T_i(y|y').$$
We denote the set of convex-product channels from $\mathcal{X}\times\mathcal{Y}'$ to $\mathcal{X}'\times\mathcal{Y}$ as $\CPC_{\mathcal{X}\times\mathcal{Y}',\mathcal{X}'\times\mathcal{Y}}$. 

\begin{myprop}
\label{propCPCCompactConvex}
The space $\CPC_{\mathcal{X}\times\mathcal{Y}',\mathcal{X}'\times\mathcal{Y}}$ is a compact and convex subset of $\DMC_{\mathcal{X}\times\mathcal{Y}',\mathcal{X}'\times\mathcal{Y}}$.
\end{myprop}
\begin{proof}
Define the set of product channels $${\PC}_{\mathcal{X}\times\mathcal{Y}',\mathcal{X}'\times\mathcal{Y}}=\{R\otimes T:\;R\in{\DMC}_{\mathcal{X},\mathcal{X}'},\;T\in{\DMC}_{\mathcal{Y}',\mathcal{Y}}\}.$$
Clearly, $\CPC_{\mathcal{X}\times\mathcal{Y}',\mathcal{X}'\times\mathcal{Y}}$ is the convex hull of $\PC_{\mathcal{X}\times\mathcal{Y}',\mathcal{X}'\times\mathcal{Y}}$ and so $\CPC_{\mathcal{X}\times\mathcal{Y}',\mathcal{X}'\times\mathcal{Y}}$ is convex. Now since $\PC_{\mathcal{X}\times\mathcal{Y}',\mathcal{X}'\times\mathcal{Y}}$ can be seen as a subset of $\mathbb{R}^{\mathcal{X}\times\mathcal{Y}'\times\mathcal{X}'\times\mathcal{Y}}$, it follows from the Carath\'{e}odory theorem that every channel $V$ in $\CPC_{\mathcal{X}\times\mathcal{Y}',\mathcal{X}'\times\mathcal{Y}}$ can be written as a convex combination of at most $$n=|\mathcal{X}\times\mathcal{Y}'\times\mathcal{X}'\times\mathcal{Y}|+1$$ product channels in $\PC_{\mathcal{X}\times\mathcal{Y}',\mathcal{X}'\times\mathcal{Y}}$. Define the mapping $$f:\Delta_{[n]}\times({\DMC}_{\mathcal{X},\mathcal{X}'}\times {\DMC}_{\mathcal{Y}',\mathcal{Y}})^n\rightarrow {\DMC}_{\mathcal{X}\times\mathcal{Y}',\mathcal{X}'\times\mathcal{Y}}$$ as
$$f\big(\alpha,(R_i,T_i)_{1\leq i\leq n}\big) = \sum_{i=1}^n \alpha(i)R_i\otimes T_i.$$
Since $\Delta_{[n]}$, ${\DMC}_{\mathcal{X},\mathcal{X}'}$ and ${\DMC}_{\mathcal{Y}',\mathcal{Y}}$ are compact, the space $\Delta_{[n]}\times({\DMC}_{\mathcal{X},\mathcal{X}'}\times {\DMC}_{\mathcal{Y}',\mathcal{Y}})^n$ is compact. Moreover, since $f$ is continuous, it follows that $${\CPC}_{\mathcal{X}\times\mathcal{Y}',\mathcal{X}'\times\mathcal{Y}}=f\big(\Delta_{[n]}\times({\DMC}_{\mathcal{X},\mathcal{X}'}\times {\DMC}_{\mathcal{Y}',\mathcal{Y}})^n\big)$$
is compact.
\end{proof}

\vspace*{3mm}

Let $\mathcal{X},\mathcal{X}',\mathcal{X}'',\mathcal{Y},\mathcal{Y}'$ and $\mathcal{Y}''$ be finite sets. For every $V\in\CPC_{\mathcal{X}\times\mathcal{Y}',\mathcal{X}'\times\mathcal{Y}}$ and every $V'\in\DMC_{\mathcal{X}'\times\mathcal{Y}'',\mathcal{X}''\times\mathcal{Y}'}$, define the \emph{skew-composition} $V\circ_s V'\in\DMC_{\mathcal{X}\times\mathcal{Y}'',\mathcal{X}''\times\mathcal{Y}}$ of $V'$ with $V$ as follows:
\begin{equation}
\label{eqSkewComposition}
(V\circ_s V')(x'',y|x,y'')=\sum_{\substack{x'\in\mathcal{X}',\\y'\in\mathcal{Y}'}}V(x',y|x,y')V'(x'',y'|x',y''),
\end{equation}
for every $x''\in\mathcal{X}''$, $y\in \mathcal{Y}$, $x\in\mathcal{X}$ and $y''\in\mathcal{Y}''$. It may not be immediately clear from \eqref{eqSkewComposition} that $V\circ_s V'$ is a valid channel in $\DMC_{\mathcal{X}\times\mathcal{Y}'',\mathcal{X}''\times\mathcal{Y}}$. In the following, we show that $V\circ_s V'\in\DMC_{\mathcal{X}\times\mathcal{Y}'',\mathcal{X}''\times\mathcal{Y}}$.

Let $n\geq 1$, $\alpha\in\Delta_{[n]}$, $(R_i,T_i)_{1\leq i\leq n}$ be such that $R_i\in\DMC_{\mathcal{X},\mathcal{X}'}$ and $T_i\in\DMC_{\mathcal{Y}',\mathcal{Y}}$ for every $1\leq i\leq n$, and
$$V=\sum_{i=1}^n\alpha(i)R_i\otimes T_i.$$
For every $(x,y'')\in\mathcal{X}\times\mathcal{Y}''$, we have
\begin{align*}
\sum_{\substack{x''\in\mathcal{X}'',\\y\in\mathcal{Y}}}(V\circ_s V')(x'',y|x,y'')&= \sum_{\substack{x''\in\mathcal{X}'',\\y\in\mathcal{Y}}}\sum_{\substack{x'\in\mathcal{X}',\\y'\in\mathcal{Y}'}}V(x',y|x,y')V'(x'',y'|x',y'')\\
&=\sum_{\substack{x''\in\mathcal{X}'',\\y\in\mathcal{Y}}}\sum_{\substack{x'\in\mathcal{X}',\\y'\in\mathcal{Y}'}}\sum_{i=1}^n\alpha(i) R_i(x'|x)T_i(y|y')V'(x'',y'|x',y'')\\
&=\sum_{i=1}^n\alpha(i) \sum_{\substack{x''\in\mathcal{X}'',\\y\in\mathcal{Y}}}\sum_{\substack{x'\in\mathcal{X}',\\y'\in\mathcal{Y}'}}R_i(x'|x)T_i(y|y')V'(x'',y'|x',y'')\\
&=\sum_{i=1}^n\alpha(i) \sum_{\substack{x''\in\mathcal{X}''}}\sum_{\substack{x'\in\mathcal{X}',\\y'\in\mathcal{Y}'}}R_i(x'|x)V'(x'',y'|x',y'')\\
&=\sum_{i=1}^n\alpha(i) \sum_{\substack{x'\in\mathcal{X}'}}R_i(x'|x)=\sum_{i=1}^n\alpha(i)=1.
\end{align*}
Therefore, $V\circ_s V' \in\DMC_{\mathcal{X}\times\mathcal{Y}'',\mathcal{X}''\times\mathcal{Y}}$. Note that if $V\in \DMC_{\mathcal{X}\times\mathcal{Y}',\mathcal{X}'\times\mathcal{Y}}$ and $V\notin  \CPC_{\mathcal{X}\times\mathcal{Y}',\mathcal{X}'\times\mathcal{Y}}$, then the skew-composition of $V'$ with $V$ as defined in Equation \eqref{eqSkewComposition} does not always yield a valid channel in $\DMC_{\mathcal{X}\times\mathcal{Y}'',\mathcal{X}''\times\mathcal{Y}}$.

\begin{mylem}
\label{lemSkewCompositionCPC}
If $V\in\CPC_{\mathcal{X}\times\mathcal{Y}',\mathcal{X}'\times\mathcal{Y}}$ and $V'\in\CPC_{\mathcal{X}'\times\mathcal{Y}'',\mathcal{X}''\times\mathcal{Y}'}$, then $V\circ_s V' \in\CPC_{\mathcal{X}\times\mathcal{Y}'',\mathcal{X}''\times\mathcal{Y}}$.
\end{mylem}
\begin{proof}
Let $n\geq 1$, $\alpha\in\Delta_{[n]}$, $(R_i,T_i)_{1\leq i\leq n}$ be such that $R_i\in\DMC_{\mathcal{X},\mathcal{X}'}$ and $T_i\in\DMC_{\mathcal{Y}',\mathcal{Y}}$ for every $1\leq i\leq n$, and
$$V=\sum_{i=1}^n\alpha(i)R_i\otimes T_i.$$
Let $n'\geq 1$, $\alpha'\in\Delta_{[n']}$, $(R_j',T_j')_{1\leq j\leq n'}$ be such that $R_j'\in\DMC_{\mathcal{X}',\mathcal{X}''}$ and $T_j'\in\DMC_{\mathcal{Y}'',\mathcal{Y}'}$ for every $1\leq j\leq n'$, and
$$V'=\sum_{j=1}^{n'}\alpha'(j)R_j'\otimes T_j'.$$
We have
\begin{align*}
(V\circ_s V')(x'',y|x,y'')&= \sum_{\substack{x'\in\mathcal{X}',\\y'\in\mathcal{Y}'}}V(x',y|x,y')V'(x'',y'|x',y'')\\
&=\sum_{\substack{x'\in\mathcal{X}',\\y'\in\mathcal{Y}'}}\sum_{i=1}^n\alpha(i)R_i(x'|x)T_i(y|y')\sum_{j=1}^{n'}\alpha'(j)R_j'(x''|x')T_j'(y'|y'')\\
&=\sum_{i=1}^n\sum_{j=1}^{n'} \alpha(i)\alpha'(j) \sum_{\substack{x'\in\mathcal{X}',\\y'\in\mathcal{Y}'}}R_i(x'|x)T_i(y|y')R_j'(x''|x')T_j'(y'|y'')\\
&=\sum_{i=1}^n\sum_{j=1}^{n'} \alpha(i)\alpha'(j) (R_j'\circ R_i)(x''|x)(T_i\circ T_j')(y|y'').
\end{align*}
Therefore, $V\circ_s V' \in\CPC_{\mathcal{X}\times\mathcal{Y}'',\mathcal{X}''\times\mathcal{Y}}$.
\end{proof}

\vspace*{3mm}

For every $W'\in\DMC_{\mathcal{X}',\mathcal{Y}'}$ and every $V\in  \CPC_{\mathcal{X}\times\mathcal{Y}',\mathcal{X}'\times\mathcal{Y}}$, we define the \emph{skew-composition} $V\circ_s W'\in\DMC_{\mathcal{X},\mathcal{Y}}$ of $W'$ with $V$ as follows:
\begin{equation}
\label{eqSkewCompositionChan}
(V\circ_s W')(y|x)=\sum_{\substack{x'\in\mathcal{X}',\\y'\in\mathcal{Y'}}}V(x',y|x,y')W'(y'|x').
\end{equation}
Note that Equation \eqref{eqSkewCompositionChan} can be seen as a particular case of Equation \eqref{eqSkewComposition} if we let $\mathcal{X}''=\mathcal{Y}''=\{0\}$ (i.e., a singleton) and we identify $\DMC_{\mathcal{X}',\mathcal{Y}'}$ with $\DMC_{\mathcal{X}'\times\mathcal{Y''},\mathcal{X}''\times\mathcal{Y}'}$.

The following lemma is trivial:
\begin{mylem}
\label{lemContainSkew}
Let $W\in\DMC_{\mathcal{X},\mathcal{Y}}$ and $W'\in\DMC_{\mathcal{X}',\mathcal{Y}'}$. $W'$ contains $W$ if and only if there exists $V\in\CPC_{\mathcal{X}\times\mathcal{Y}',\mathcal{X}'\times\mathcal{Y}}$ such that $W=V\circ_s W'$.
\end{mylem}

\section{A characterization of the Shannon ordering}

A \emph{blind randomized in the middle (BRM) game} is a 6-tuple $\mathcal{G}= (\mathcal{U},\mathcal{X},\mathcal{Y},\mathcal{V}, l,W)$ such that $\mathcal{U},\mathcal{X},\mathcal{Y}$ and $\mathcal{V}$ are finite sets, $l$ is a mapping from $\mathcal{U}\times\mathcal{V}$ to $\mathbb{R}$, and $W\in\DMC_{\mathcal{X},\mathcal{Y}}$. The mapping $l$ is called the \emph{payoff function} of the BRM game $\mathcal{G}$, and the channel $W$ is called the \emph{randomizer} of $\mathcal{G}$. The BRM game consists of two players that we call Alice and Bob. The BRM game takes place in two stages:
\begin{itemize}
\item Alice chooses a symbol $u\in\mathcal{U}$ and writes her choice on a piece of paper. Bob chooses two functions $f:\mathcal{U}\rightarrow\mathcal{X}$ and $g:\mathcal{Y}\rightarrow\mathcal{V}$, and writes a description of $f$ and $g$ on a piece of paper. At this stage, no player has knowledge of the choice of the other player.
\item Alice and Bob simultaneously reveal their papers. They compute $x=f(u)\in\mathcal{X}$ and then randomly generate a symbol $y\in\mathcal{Y}$ according to the conditional probability distribution $W(y|x)$. Finally, $v=g(y)$ is computed and then Alice pays\footnote{If $l(u,v)<0$, then Bob pays Alice an amount of money that is equal to $-l(u,v)$.} Bob an amount of money that is equal to $l(u,v)$.
\end{itemize}

A \emph{strategy} (for Bob) in the BRM game $\mathcal{G}$ is a 4-tuple $S=(n,\alpha, \mathbf{f}, \mathbf{g})$ satisfying:
\begin{itemize}
\item $n\geq 1$ is a strictly positive integer.
\item $\alpha\in\Delta_{[n]}$.
\item $\mathbf{f}=(f_i)_{1\leq i\leq n}\in (\mathcal{X}^{\mathcal{U}})^n$, where $\mathcal{X}^{\mathcal{U}}$ is the set of functions from $\mathcal{U}$ to $\mathcal{X}$.
\item $\mathbf{g}=(g_i)_{1\leq i\leq n}\in (\mathcal{V}^{\mathcal{Y}})^n$.
\end{itemize}
We denote $n$ and $\alpha$ as $n_S$ and $\alpha_S$ respectively. For every $1\leq i\leq n=n_S$, we denote $f_i$ and $g_i$ as $f_{i,S}$ and $g_{i,S}$ respectively. The set of strategies is denoted as $\mathcal{S}_{\mathcal{U},\mathcal{X},\mathcal{Y},\mathcal{V}}$.

Bob implements the strategy $S$ as follows: he randomly picks an index $i\in\{1,\ldots,n_S\}$ according to the distribution $\alpha_S$, and then commits to the choice $(f_{i,S},g_{i,S})$.

For every $u\in\mathcal{U}$, the \emph{payoff gained by the strategy $S$ for $u$ in the BRM game $\mathcal{G}$} is given by:
$$\$(u,S,\mathcal{G})=\sum_{i=1}^{n_S} \alpha_S(i) \sum_{y\in\mathcal{Y}} W(y|f_{i,S}(u))l(u,g_{i,S}(y)).$$
The \emph{payoff vector gained by the strategy $S$ in the game $\mathcal{G}$} is given by:
$$\vec{\$}(S,\mathcal{G})=\big(\$(u,S,\mathcal{G})\big)_{u\in\mathcal{U}}\in\mathbb{R}^{\mathcal{U}}.$$

The \emph{achievable payoff region for the game $\mathcal{G}$} is given by:
$$\$_{\ach}(\mathcal{G})=\Big\{\vec{\$}(S,\mathcal{G}):\; S\in\mathcal{S}_{\mathcal{U},\mathcal{X},\mathcal{Y},\mathcal{V}}\Big\}\subset \mathbb{R}^{\mathcal{U}}.$$

The \emph{average payoff for the strategy $S\in\mathcal{S}_{\mathcal{U},\mathcal{X},\mathcal{Y},\mathcal{V}}$ in the game $\mathcal{G}$} is given by:
$$\hat{\$}(S,\mathcal{G})=\frac{1}{|\mathcal{U}|}\sum_{u\in\mathcal{U}} \$(u,S,\mathcal{G}).$$
$\hat{\$}(S,\mathcal{G})$ is the expected gain of Bob assuming that Alice chooses $u\in\mathcal{U}$ uniformly at random.

The \emph{optimal average payoff for the game $\mathcal{G}$} is given by
$$\$_{\opt}(\mathcal{G})= \sup_{S\in\mathcal{S}_{\mathcal{U},\mathcal{X},\mathcal{Y},\mathcal{V}}} \hat{\$}(S,\mathcal{G}).$$

For every $S\in\mathcal{S}_{\mathcal{U},\mathcal{X},\mathcal{Y},\mathcal{V}}$, we associate the convex-product channel $V_S\in\CPC_{\mathcal{U}\times\mathcal{Y},\mathcal{X}\times\mathcal{V}}$ defined as
$$V_S=\sum_{i=1}^{n_S}\alpha_S(i)D_{f_{i,S}}\otimes D_{g_{i,S}}.$$
For every $u\in\mathcal{U}$, we have
\begin{equation}
\label{eqDollarVs}
\begin{aligned}
\$(u,S,\mathcal{G})&=\sum_{i=1}^{n_S} \alpha_S(i) \sum_{y\in\mathcal{Y}} W(y|f_{i,S}(u))l(u,g_{i,S}(y))\\
&=\sum_{i=1}^{n_S} \alpha_S(i) \sum_{\substack{x\in\mathcal{X},\\y\in\mathcal{Y},\\v\in\mathcal{V}}} D_{f_{i,S}}(x|u) W(y|x) D_{g_{i,S}}(v|y) l(u,v)\\
&=\sum_{\substack{x\in\mathcal{X},\\y\in\mathcal{Y},\\v\in\mathcal{V}}} \left(\sum_{i=1}^{n_S} \alpha_S(i)  D_{f_{i,S}}(x|u)D_{g_{i,S}}(v|y)\right) W(y|x) l(u,v)\\
&=\sum_{\substack{x\in\mathcal{X},\\y\in\mathcal{Y},\\v\in\mathcal{V}}} V_S(x,v|u,y) W(y|x) l(u,v).
\end{aligned}
\end{equation}

\begin{mylem}
\label{lemCPCisStrategy}
For every $V\in\CPC_{\mathcal{U}\times\mathcal{Y},\mathcal{X}\times\mathcal{V}}$, there exists $S\in\mathcal{S}_{\mathcal{U},\mathcal{X},\mathcal{Y},\mathcal{V}}$ such that $V=V_S$.
\end{mylem}
\begin{proof}
Let $n\geq 1$, $\alpha\in\Delta_{[n]}$, $(R_i,T_i)_{1\leq i\leq n}$ be such that $R_i\in\DMC_{\mathcal{U},\mathcal{X}}$ and $T_i\in\DMC_{\mathcal{Y},\mathcal{V}}$ for every $1\leq i\leq n$, and
\begin{equation}
\label{eqVVVVVVV}
V=\sum_{i=1}^n\alpha(i)R_i\otimes T_i.
\end{equation}
Since every channel can be written as a convex combination of deterministic channels \cite{ShannonDegrad}, we can rewrite \eqref{eqVVVVVVV} as a convex combination of products of deterministic channels. Therefore, there exists $S\in\mathcal{S}_{\mathcal{U},\mathcal{X},\mathcal{Y},\mathcal{V}}$ such that $V=V_S$.
\end{proof}

\vspace*{3mm}

Equation \eqref{eqDollarVs} and Lemma \ref{lemCPCisStrategy} imply that $\$_{\ach}(\mathcal{G})$ is the image of $\CPC_{\mathcal{U}\times\mathcal{Y},\mathcal{X}\times\mathcal{V}}$ by a linear function. Since $\CPC_{\mathcal{U}\times\mathcal{Y},\mathcal{X}\times\mathcal{V}}$ is convex and compact (Proposition \ref{propCPCCompactConvex}), $\$_{\ach}(\mathcal{G})$ is convex and compact as well.

Let $\mathcal{U}$ and $\mathcal{V}$ be two finite sets and let $l:\mathcal{U}\times\mathcal{V}\rightarrow\mathbb{R}$ be a payoff function. We say that $l$ is \emph{normalized and positive} if $l(u,v)\geq 0$ for every $u\in\mathcal{U}$ and every $v\in\mathcal{V}$, and $$\sum_{\substack{u\in\mathcal{U},\\v\in\mathcal{V}}}l(u,v)=1.$$
In other words, $l$ is normalized and positive if $l\in\Delta_{\mathcal{U}\times\mathcal{V}}$.

The following theorem provides a characterization of the Shannon ordering of communication channels that is similar to the BSS theorem.
\begin{mythe}
\label{theGameShannonCharac}
Let $\mathcal{X},\mathcal{X}',\mathcal{Y}$ and $\mathcal{Y}'$ be four finite sets. Let $W\in\DMC_{\mathcal{X},\mathcal{Y}}$ and $W'\in\DMC_{\mathcal{X}',\mathcal{Y}'}$. The following conditions are equivalent:
\begin{itemize}
\item[(a)] $W'$ contains $W$.
\item[(b)] For every two finite sets $\mathcal{U}$ and $\mathcal{V}$, and every payoff function $l:\mathcal{U}\times\mathcal{V}\rightarrow\mathbb{R}$, we have
$$\$_{\ach}(\mathcal{U},\mathcal{X},\mathcal{Y},\mathcal{V},l,W)\subset \$_{\ach}(\mathcal{U},\mathcal{X}',\mathcal{Y}',\mathcal{V},l,W').$$
\item[(c)] For every two finite sets $\mathcal{U}$ and $\mathcal{V}$, and every payoff function $l:\mathcal{U}\times\mathcal{V}\rightarrow\mathbb{R}$, we have
$$\$_{\opt}(\mathcal{U},\mathcal{X},\mathcal{Y},\mathcal{V},l,W)\leq \$_{\opt}(\mathcal{U},\mathcal{X}',\mathcal{Y}',\mathcal{V},l,W').$$
\item[(d)] For every two finite sets $\mathcal{U}$ and $\mathcal{V}$, and every normalized and positive payoff function $l\in\Delta_{\mathcal{U}\times\mathcal{V}}$, we have
$$\$_{\ach}(\mathcal{U},\mathcal{X},\mathcal{Y},\mathcal{V},l,W)\subset \$_{\ach}(\mathcal{U},\mathcal{X}',\mathcal{Y}',\mathcal{V},l,W').$$
\item[(e)] For every two finite sets $\mathcal{U}$ and $\mathcal{V}$, and every normalized and positive payoff function $l\in\Delta_{\mathcal{U}\times\mathcal{V}}$, we have
$$\$_{\opt}(\mathcal{U},\mathcal{X},\mathcal{Y},\mathcal{V},l,W)\leq \$_{\opt}(\mathcal{U},\mathcal{X}',\mathcal{Y}',\mathcal{V},l,W').$$
\end{itemize}
\end{mythe}
\begin{proof}
Assume that (a) is true. Lemma \ref{lemContainSkew} implies that there exists $V\in \CPC_{\mathcal{X}\times\mathcal{Y}',\mathcal{X}'\times\mathcal{Y}}$ such that $W=V\circ_s W'$. Let $\mathcal{U}$ and $\mathcal{V}$ be two finite sets, and let $l:\mathcal{U}\times\mathcal{V}\rightarrow\mathbb{R}$ be a payoff function. Define $\mathcal{G}=(\mathcal{U},\mathcal{X},\mathcal{Y},\mathcal{V},l,W)$ and $\mathcal{G}'=(\mathcal{U},\mathcal{X}',\mathcal{Y}',\mathcal{V},l,W')$.

Fix $\vec{v}\in \$_{\ach}(\mathcal{G})$. There exists $S\in\mathcal{S}_{\mathcal{U},\mathcal{X},\mathcal{Y},\mathcal{V}}$ such that $\vec{v}=\vec{\$}(S,\mathcal{G})=\big(\$(u,S,\mathcal{G})\big)_{u\in\mathcal{U}}$. From equation \eqref{eqDollarVs} we have:
\begin{align*}
\$(u,S,\mathcal{G})&=\sum_{\substack{x\in\mathcal{X},\\y\in\mathcal{Y},\\v\in\mathcal{V}}} V_S(x,v|u,y) W(y|x) l(u,v)\\
&=\sum_{\substack{x\in\mathcal{X},\\y\in\mathcal{Y},\\v\in\mathcal{V}}}  V_S(x,v|u,y) \Bigg(\sum_{\substack{x'\in\mathcal{X}',\\y'\in\mathcal{Y'}}}V(x',y|x,y')W'(y'|x')\Bigg) l(u,v)\\
&=\sum_{\substack{x'\in\mathcal{X}',\\y'\in\mathcal{Y'},\\v\in\mathcal{V}}} \Bigg(\sum_{\substack{x\in\mathcal{X},\\y\in\mathcal{Y}}} V_S(x,v|u,y)V(x',y|x,y')\Bigg)W'(y'|x') l(u,v)\\
&=\sum_{\substack{x'\in\mathcal{X}',\\y'\in\mathcal{Y'},\\v\in\mathcal{V}}} (V_S\circ_s V)(x',v|u,y')W'(y'|x') l(u,v).
\end{align*}
Lemma \ref{lemSkewCompositionCPC} implies that $V_S\circ_s V\in\CPC_{\mathcal{U}\times\mathcal{Y}',\mathcal{X}'\times\mathcal{V}}$ and Lemma \ref{lemCPCisStrategy} implies that there exists $S'\in\mathcal{S}_{\mathcal{U},\mathcal{X}',\mathcal{Y}',\mathcal{V}}$ such that $V_{S'}=V_S\circ_s V$. Therefore,
\begin{align*}
\$(u,S,\mathcal{G})&=\sum_{\substack{x'\in\mathcal{X}',\\y'\in\mathcal{Y'},\\v\in\mathcal{V}}} V_{S'}(x',v|u,y')W'(y'|x') l(u,v)\stackrel{(\ast)}{=}\$(u,S',\mathcal{G}'),
\end{align*}
where $(\ast)$ follows from Equation \eqref{eqDollarVs}. This shows that $\vec{v}=\big(\$(u,S',\mathcal{G}')\big)_{u\in\mathcal{U}}$, hence $\$_{\ach}(\mathcal{G})\subset \$_{\ach}(\mathcal{G}')$. Therefore, (a) implies (b).

Now assume that (b) is true. Let $\mathcal{U}$ and $\mathcal{V}$ be two finite sets, and let $l:\mathcal{U}\times\mathcal{V}\rightarrow\mathbb{R}$ be a payoff function. Define $\mathcal{G}=(\mathcal{U},\mathcal{X},\mathcal{Y},\mathcal{V},l,W)$ and $\mathcal{G}'=(\mathcal{U},\mathcal{X}',\mathcal{Y}',\mathcal{V},l,W')$. We have $\$_{\ach}(\mathcal{G})\subset \$_{\ach}(\mathcal{G}')$. Therefore,
\begin{align*}
\$_{\opt}(\mathcal{G})=\sup_{(v_u)_{u\in\mathcal{U}} \in \$_{\ach}(\mathcal{G})} \frac{1}{|\mathcal{U}|}\sum_{u\in\mathcal{U}} v_u \stackrel{(\ast\ast)}{\leq} \sup_{(v_u')_{u\in\mathcal{U}} \in \$_{\ach}(\mathcal{G}')} \frac{1}{|\mathcal{U}|} \sum_{u\in\mathcal{U}} v_u' = \$_{\opt}(\mathcal{G}'),
\end{align*}
where $(\ast\ast)$ follows from the fact that $\$_{\ach}(\mathcal{G})\subset \$_{\ach}(\mathcal{G}')$. This shows that (b) implies (c). We can show similarly that (d) implies (e).

Trivially, (b) implies (d), and (c) implies (e).

Now assume that (e) is true. For every normalized and positive payoff function $l\in\Delta_{\mathcal{X}\times\mathcal{Y}}$, define the BRM games $\mathcal{G}=(\mathcal{X},\mathcal{X},\mathcal{Y},\mathcal{Y},l,W)$ and $\mathcal{G}'=(\mathcal{X},\mathcal{X}',\mathcal{Y}',\mathcal{Y},l,W')$. We have $\$_{\opt}(\mathcal{G})\leq \$_{\opt}(\mathcal{G}')$.

Fix a strategy $S\in\mathcal{S}_{\mathcal{X},\mathcal{X},\mathcal{Y},\mathcal{Y}}$ satisfying $n_S=1$, $f_{1,S}(x)=x$ for all $x\in\mathcal{X}$ and $g_{1,S}(y)=y$ for all $y\in\mathcal{Y}$. Clearly $\alpha_S(1)=1$, hence
\begin{align*}
\hat{\$}(S,\mathcal{G})&=\frac{1}{|\mathcal{X}|}\sum_{x\in\mathcal{X}} \$(x,S,\mathcal{G})=\frac{1}{|\mathcal{X}|}\sum_{x\in\mathcal{X}} \sum_{y\in\mathcal{Y}}W(y|f_{1,S}(x)) l\big(x,g_{1,S}(y)\big)=\frac{1}{|\mathcal{X}|}\sum_{\substack{x\in\mathcal{X},\\y\in\mathcal{Y}}} W(y|x)l(x,y).
\end{align*}
Therefore,
\begin{align*}
\frac{1}{|\mathcal{X}|}\sum_{\substack{x\in\mathcal{X},\\y\in\mathcal{Y}}} W(y|x)l(x,y)&=\hat{\$}(S,\mathcal{G})\leq \$_{\opt}(\mathcal{G})\leq \$_{\opt}(\mathcal{G}')=\sup_{S'\in\mathcal{S}_{\mathcal{X},\mathcal{X}',\mathcal{Y}',\mathcal{Y}}} \hat{\$}(S',\mathcal{G}')\\
&=\sup_{S'\in\mathcal{S}_{\mathcal{X},\mathcal{X}',\mathcal{Y}',\mathcal{Y}}} \frac{1}{|\mathcal{X}|}\sum_{\substack{x\in\mathcal{X}}} \$(x,S',\mathcal{G}')\\
&=\sup_{S'\in\mathcal{S}_{\mathcal{X},\mathcal{X}',\mathcal{Y}',\mathcal{Y}}} \frac{1}{|\mathcal{X}|}\sum_{\substack{x\in\mathcal{X}}}\sum_{\substack{x'\in\mathcal{X}',\\y'\in\mathcal{Y}',\\y\in\mathcal{Y}}} V_{S'}(x',y|x,y') W'(y'|x') l(x,y)\\
&=\sup_{S'\in\mathcal{S}_{\mathcal{X},\mathcal{X}',\mathcal{Y}',\mathcal{Y}}} \frac{1}{|\mathcal{X}|}\sum_{\substack{x\in\mathcal{X},\\y\in\mathcal{Y}}} (V_{S'}\circ_s W')(y|x) l(x,y)\\
&\stackrel{(\dagger)}{=}\sup_{V\in\CPC_{\mathcal{X}\times\mathcal{Y}',\mathcal{X}'\times\mathcal{Y}}} \frac{1}{|\mathcal{X}|}\sum_{\substack{x\in\mathcal{X},\\y\in\mathcal{Y}}} (V\circ_s W')(y|x) l(x,y),
\end{align*}
where $(\dagger)$ follows from Lemma \ref{lemCPCisStrategy}. Therefore,
\begin{align*}
\inf_{V\in\CPC_{\mathcal{X}\times\mathcal{Y}',\mathcal{X}'\times\mathcal{Y}}} \frac{1}{|\mathcal{X}|}\sum_{\substack{x\in\mathcal{X},\\y\in\mathcal{Y}}}\big(W(y|x) - (V\circ_s W')(y|x)\big)l(x,y)\leq 0.
\end{align*}
Since this is true for every $l\in\Delta_{\mathcal{X}\times\mathcal{Y}}$, we have:
\begin{align*}
\sup_{l\in\Delta_{\mathcal{X}\times \mathcal{Y}}} \inf_{V\in\CPC_{\mathcal{X}\times\mathcal{Y}',\mathcal{X}'\times\mathcal{Y}}} \sum_{\substack{x\in\mathcal{X},\\y\in\mathcal{Y}}}\big(W(y|x) - (V\circ_s W')(y|x)\big)l(x,y)\leq 0.
\end{align*}
Moreover, since $\Delta_{\mathcal{X}\times \mathcal{Y}}$ and $\CPC_{\mathcal{X}\times\mathcal{Y}',\mathcal{X}'\times\mathcal{Y}}$ are compact (see Proposition \ref{propCPCCompactConvex}), the sup and the inf are attainable. Therefore, we can write:
\begin{equation}
\label{eqMinimaxlV}
\max_{l\in\Delta_{\mathcal{X}\times \mathcal{Y}}} \min_{V\in\CPC_{\mathcal{X}\times\mathcal{Y}',\mathcal{X}'\times\mathcal{Y}}} \sum_{\substack{x\in\mathcal{X},\\y\in\mathcal{Y}}}\big(W(y|x) - (V\circ_s W')(y|x)\big)l(x,y)\leq 0.
\end{equation}
Since the function $\displaystyle \sum_{\substack{x\in\mathcal{X},\\y\in\mathcal{Y}}}\big(W(y|x) - (V\circ_s W')(y|x)\big)l(x,y)$ is affine in both $l\in\Delta_{\mathcal{X}\times \mathcal{Y}}$ and $V\in\CPC_{\mathcal{X}\times\mathcal{Y}',\mathcal{X}'\times\mathcal{Y}}$, it is continuous, concave in $l$ and convex in $V$. On the other hand, the sets $\Delta_{\mathcal{X}\times \mathcal{Y}}$ and $\CPC_{\mathcal{X}\times\mathcal{Y}',\mathcal{X}'\times\mathcal{Y}}$ are compact and convex (see Proposition \ref{propCPCCompactConvex}). Therefore, we can apply the minimax theorem \cite{MiniMax} to exchange the max and the min in Equation \eqref{eqMinimaxlV}. We obtain:
\begin{align*}
\min_{V\in\CPC_{\mathcal{X}\times\mathcal{Y}',\mathcal{X}'\times\mathcal{Y}}} \max_{l\in\Delta_{\mathcal{X}\times \mathcal{Y}}} \sum_{\substack{x\in\mathcal{X},\\y\in\mathcal{Y}}}\big(W(y|x) - (V\circ_s W')(y|x)\big)l(x,y)\leq 0.
\end{align*}
Therefore, there exists $V\in\CPC_{\mathcal{X}\times\mathcal{Y}',\mathcal{X}'\times\mathcal{Y}}$ such that
\begin{align*}
0&\geq \max_{l\in\Delta_{\mathcal{X}\times \mathcal{Y}}} \sum_{\substack{x\in\mathcal{X},\\y\in\mathcal{Y}}}\big(W(y|x) - (V\circ_s W')(y|x)\big)l(x,y)\\
&\stackrel{(\dagger\dagger)}{=}\max_{\substack{x\in\mathcal{X},\\y\in\mathcal{Y}}} \big(W(y|x) - (V\circ_s W')(y|x)\big),
\end{align*}
where $(\dagger\dagger)$ follows from the fact that $\displaystyle\sum_{\substack{x\in\mathcal{X},\\y\in\mathcal{Y}}}\big(W(y|x) - (V\circ_s W')(y|x)\big)l(x,y)$ is maximized when we choose $l\in\Delta_{\mathcal{X},\mathcal{Y}}$ in such a way that $l(x_0,y_0)=1$ for any $(x_0,y_0)\in\mathcal{X}\times\mathcal{Y}$ satisfying
$$\big(W(y_0|x_0) - (V\circ_s W')(y_0|x_0)\big)=\max_{\substack{x\in\mathcal{X},\\y\in\mathcal{Y}}} \big(W(y|x) - (V\circ_s W')(y|x)\big).$$
We conclude that for every $(x,y)\in\mathcal{X}\times\mathcal{Y}$, we have
$$W(y|x)\leq (V\circ_s W')(y|x).$$
Now since $\displaystyle \sum_{y\in\mathcal{Y}} W(y|x)=\sum_{y\in\mathcal{Y}} (V\circ_s W')(y|x)$ for every $x\in\mathcal{X}$, we must have $W(y|x)= (V\circ_s W')(y|x)$ for every $(x,y)\in\mathcal{X}\times\mathcal{Y}$. Therefore, $W=V\circ_s W'$. Lemma \ref{lemContainSkew} now implies that $W'$ contains $W$, hence (e) implies (a). We conclude that the conditions (a), (b), (c), (d) and (e) are equivalent.
\end{proof}

\section{Space of Shannon-equivalent channels from $\mathcal{X}$ to $\mathcal{Y}$}

\subsection{The $\DMC_{\mathcal{X},\mathcal{Y}}^{(s)}$ space}

\label{subsecDMCXYs}
Let $\mathcal{X}$ and $\mathcal{Y}$ be two finite sets. Define the equivalence relation $R_{\mathcal{X},\mathcal{Y}}^{(s)}$ on $\DMC_{\mathcal{X},\mathcal{Y}}$ as follows:
$$WR_{\mathcal{X},\mathcal{Y}}^{(s)}W'\;\;\Leftrightarrow\;\;W\;\text{is Shannon-equivalent to}\;W'.$$

\begin{mydef}
The space of Shannon-equivalent channels with input alphabet $\mathcal{X}$ and output alphabet $\mathcal{Y}$ is the quotient of the space of channels from $\mathcal{X}$ to $\mathcal{Y}$ by the Shannon-equivalence relation:
$$\textstyle\DMC_{\mathcal{X},\mathcal{Y}}^{(s)}=\DMC_{\mathcal{X},\mathcal{Y}}/R_{\mathcal{X},\mathcal{Y}}^{(s)}.$$
We define the topology $\mathcal{T}_{\mathcal{X},\mathcal{Y}}^{(s)}$ on $\DMC_{\mathcal{X},\mathcal{Y}}^{(s)}$ as the quotient topology $\mathcal{T}_{\mathcal{X},\mathcal{Y}}/R_{\mathcal{X},\mathcal{Y}}^{(s)}$.
\end{mydef}

\vspace*{3mm}

\begin{mynot}
Let $(\mathcal{U},\mathcal{X},\mathcal{Y},\mathcal{V},l,W)$ be a BRM game. Since $\mathcal{U},\mathcal{X},\mathcal{Y}$ and $\mathcal{V}$ are implicitly determined by $l$ and $W$, we may simply write $\$_{\opt}(l,W)$ to denote $\$_{\opt}(\mathcal{U},\mathcal{X},\mathcal{Y},\mathcal{V},l,W)$.

Let $W,W'\in\DMC_{\mathcal{X},\mathcal{Y}}$. Theorem \ref{theGameShannonCharac} shows that $W'$ contains $W$ if and only if $\$_{\opt}(l,W)\leq \$_{\opt}(l,W')$ for every $l\in\Delta_{\mathcal{U}\times\mathcal{V}}$ and every two finite sets $\mathcal{U}$ and $\mathcal{V}$. Therefore, $W R_{\mathcal{X},\mathcal{Y}}^{(s)}W'$ if and only if $\$_{\opt}(l,W)= \$_{\opt}(l,W')$ for every $l\in\Delta_{\mathcal{U}\times\mathcal{V}}$ and every two finite sets $\mathcal{U}$ and $\mathcal{V}$. This shows that $\$_{\opt}(l,W)$ only depends on the $R_{\mathcal{X},\mathcal{Y}}^{(s)}$-equivalence class of $W$. Therefore, if $\hat{W}\in\DMC_{\mathcal{X},\mathcal{Y}}^{(s)}$, we can define $\$_{\opt}(l,\hat{W}):=\$_{\opt}(l,W')$ for any $W'\in\hat{W}$.
\end{mynot}

Define the \emph{BRM metric} $d_{\mathcal{X},\mathcal{Y}}^{(s)}$ on $\DMC_{\mathcal{X},\mathcal{Y}}^{(s)}$ as follows:
$$d_{\mathcal{X},\mathcal{Y}}^{(s)}(\hat{W}_1,\hat{W}_2)=\sup_{\substack{n,m\geq 1,\\l\in{\Delta}_{[n]\times[m]}}}|\$_{\opt}(l,\hat{W}_1)-\$_{\opt}(l,\hat{W}_2)|.$$

\begin{myprop}
\label{propReldXYdXYs}
Let $W_1,W_2\in\DMC_{\mathcal{X},\mathcal{Y}}$ and let $\hat{W}_1$ and $\hat{W}_2$ be the $R_{\mathcal{X},\mathcal{Y}}^{(s)}$-equivalence classes of $W_1$ and $W_2$ respectively. We have $d_{\mathcal{X},\mathcal{Y}}^{(s)}(\hat{W}_1,\hat{W}_2)\leq d_{\mathcal{X},\mathcal{Y}}(W_1,W_2)$.
\end{myprop}
\begin{proof}
See Appendix \ref{appReldXYdXYs}.
\end{proof}

\begin{mythe}
\label{theDMCXYs}
The topology induced by $d_{\mathcal{X},\mathcal{Y}}^{(s)}$ on $\DMC_{\mathcal{X},\mathcal{Y}}^{(s)}$ is the same as the quotient topology $\mathcal{T}_{\mathcal{X},\mathcal{Y}}^{(s)}$. Moreover, $(\DMC_{\mathcal{X},\mathcal{Y}}^{(s)},d_{\mathcal{X},\mathcal{Y}}^{(s)})$ is compact and path-connected.
\end{mythe}
\begin{proof}
Since $(\DMC_{\mathcal{X},\mathcal{Y}},d_{\mathcal{X},\mathcal{Y}})$ is compact and path-connected, the quotient space $(\DMC_{\mathcal{X},\mathcal{Y}}^{(s)},\mathcal{T}_{\mathcal{X},\mathcal{Y}}^{(s)})$ is compact and path-connected.

Define the mapping $\Proj:\DMC_{\mathcal{X},\mathcal{Y}}\rightarrow\DMC_{\mathcal{X},\mathcal{Y}}^{(s)}$ as $\Proj(W)=\hat{W}$, where $\hat{W}$ is the $R_{\mathcal{X},\mathcal{Y}}^{(s)}$-equivalence class of $W$. Proposition \ref{propReldXYdXYs} implies that $\Proj$ is a continuous mapping from $(\DMC_{\mathcal{X},\mathcal{Y}},d_{\mathcal{X},\mathcal{Y}})$ to $(\DMC_{\mathcal{X},\mathcal{Y}}^{(s)},d_{\mathcal{X},\mathcal{Y}}^{(s)})$. Since $\Proj(W)$ depends only on $\hat{W}$, Lemma \ref{lemQuotientFunction} implies that the transcendent mapping of $\Proj$ defined on the quotient space $(\DMC_{\mathcal{X},\mathcal{Y}}^{(s)},\mathcal{T}_{\mathcal{X},\mathcal{Y}}^{(s)})$ is continuous. But the transcendent mapping of $\Proj$ is nothing but the identity on $\DMC_{\mathcal{X},\mathcal{Y}}^{(s)}$. Therefore, the identity mapping $id$ on $\DMC_{\mathcal{X},\mathcal{Y}}^{(s)}$ is a continuous mapping from $(\DMC_{\mathcal{X},\mathcal{Y}}^{(s)},\mathcal{T}_{\mathcal{X},\mathcal{Y}}^{(s)})$ to $(\DMC_{\mathcal{X},\mathcal{Y}}^{(s)},d_{\mathcal{X},\mathcal{Y}}^{(s)})$.
For every subset $U$ of $\DMC_{\mathcal{X},\mathcal{Y}}^{(s)}$ we have:
\begin{itemize}
\item If $U$ is open in $(\DMC_{\mathcal{X},\mathcal{Y}}^{(s)},d_{\mathcal{X},\mathcal{Y}}^{(s)})$, then $U=id^{-1}(U)$ is open in $(\DMC_{\mathcal{X},\mathcal{Y}}^{(s)},\mathcal{T}_{\mathcal{X},\mathcal{Y}}^{(s)})$.
\item If $U$ is open in $(\DMC_{\mathcal{X},\mathcal{Y}}^{(s)},\mathcal{T}_{\mathcal{X},\mathcal{Y}}^{(s)})$, then its complement $U^c$ is closed in $(\DMC_{\mathcal{X},\mathcal{Y}}^{(s)},\mathcal{T}_{\mathcal{X},\mathcal{Y}}^{(s)})$ which is compact, hence $U^c$ is compact in $(\DMC_{\mathcal{X},\mathcal{Y}}^{(s)},\mathcal{T}_{\mathcal{X},\mathcal{Y}}^{(s)})$. This shows that $U^c=id(U^c)$ is a compact subset of $(\DMC_{\mathcal{X},\mathcal{Y}}^{(s)},d_{\mathcal{X},\mathcal{Y}}^{(s)})$. But $(\DMC_{\mathcal{X},\mathcal{Y}}^{(s)},d_{\mathcal{X},\mathcal{Y}}^{(s)})$ is a metric space, so $U^c$ is closed in $(\DMC_{\mathcal{X},\mathcal{Y}}^{(s)},d_{\mathcal{X},\mathcal{Y}}^{(s)})$. Therefore, $U$ is open $(\DMC_{\mathcal{X},\mathcal{Y}}^{(s)},d_{\mathcal{X},\mathcal{Y}}^{(s)})$.
\end{itemize}
We conclude that $(\DMC_{\mathcal{X},\mathcal{Y}}^{(s)},\mathcal{T}_{\mathcal{X},\mathcal{Y}}^{(s)})$ and $(\DMC_{\mathcal{X},\mathcal{Y}}^{(s)},d_{\mathcal{X},\mathcal{Y}}^{(s)})$ have the same open sets. Therefore, the topology induced by $d_{\mathcal{X},\mathcal{Y}}^{(s)}$ on $\DMC_{\mathcal{X},\mathcal{Y}}^{(s)}$ is the same as the quotient topology $\mathcal{T}_{\mathcal{X},\mathcal{Y}}^{(s)}$. Now since $(\DMC_{\mathcal{X},\mathcal{Y}}^{(s)},\mathcal{T}_{\mathcal{X},\mathcal{Y}}^{(s)})$ is compact and path-connected, $(\DMC_{\mathcal{X},\mathcal{Y}}^{(s)},d_{\mathcal{X},\mathcal{Y}}^{(s)})$ is compact and path-connected as well.
\end{proof}

\vspace*{3mm}

Throughout this paper, we always associate $\DMC_{\mathcal{X},\mathcal{Y}}^{(s)}$ with the BRM metric $d_{\mathcal{X},\mathcal{Y}}^{(s)}$ and the quotient topology $\mathcal{T}_{\mathcal{X},\mathcal{Y}}^{(s)}$.

\subsection{Canonical embedding and canonical identification}

\label{subsecEmbedShanEquiv}

Let $\mathcal{X}_1,\mathcal{X}_2,\mathcal{Y}_1$ and $\mathcal{Y}_2$ be four finite sets such that $|\mathcal{X}_1|\leq |\mathcal{X}_2|$ and $|\mathcal{Y}_1|\leq |\mathcal{Y}_2|$. We will show that there is a canonical embedding from $\DMC_{\mathcal{X}_1,\mathcal{Y}_1}^{(s)}$ to $\DMC_{\mathcal{X}_2,\mathcal{Y}_2}^{(s)}$. In other words, there exists an explicitly constructable compact subset $A$ of $\DMC_{\mathcal{X}_2,\mathcal{Y}_2}^{(s)}$ such that $A$ is homeomorphic to $\DMC_{\mathcal{X}_1,\mathcal{Y}_1}^{(s)}$. $A$ and the homeomorphism depend only on $\mathcal{X}_1,\mathcal{X}_2,\mathcal{Y}_1$ and $\mathcal{Y}_2$ (this is why we say that they are canonical). Moreover, we can show that $A$ depends only on $|\mathcal{X}_1|$, $|\mathcal{Y}_1|$, $\mathcal{X}_2$ and $\mathcal{Y}_2$.

\begin{mylem}
\label{lemEquivChannelSurjInj}
For every $W\in\DMC_{\mathcal{X}_1,\mathcal{Y}_1}$, every surjection $f$ from $\mathcal{X}_2$ to $\mathcal{X}_1$, and  every injection $g$ from $\mathcal{Y}_1$ to $\mathcal{Y}_2$, the channel $W$ is Shannon-equivalent to $D_g\circ W\circ D_f$.
\end{mylem}
\begin{proof}
Clearly $W$ contains $D_g\circ W\circ D_f$. Now let $f'$ be any mapping from $\mathcal{X}_1$ to $\mathcal{X}_2$ such that $f(f'(x_1))=x_1$ for every $x_1\in\mathcal{X}_1$, and let $g'$ be any mapping from $\mathcal{Y}_2$ to $\mathcal{Y}_1$ such that $g'(g(y_1))=y_1$ for every $y_1\in\mathcal{Y}_1$. We have $$W=(D_{g'}\circ D_g)\circ W\circ (D_{f}\circ D_{f'})=D_{g'}\circ(D_g\circ W\circ D_f)\circ D_{f'},$$ and so $D_g\circ W\circ D_f$ also contains $W$. Therefore, $W$ and $D_g\circ W\circ D_f$ are Shannon-equivalent.
\end{proof}

\begin{mycor}
\label{corEquivChannelSurjInj}
For every $W,W'\in\DMC_{\mathcal{X}_1,\mathcal{Y}_1}$, every two surjections $f,f'$ from $\mathcal{X}_2$ to $\mathcal{X}_1$, and every two injections $g,g'$ from $\mathcal{Y}_1$ to $\mathcal{Y}_2$, we have:
$$W R_{\mathcal{X}_1,\mathcal{Y}_1}^{(s)} W'\;\;\Leftrightarrow\;\; (D_g\circ W\circ D_f)R_{\mathcal{X}_2,\mathcal{Y}_2}^{(s)}(D_{g'}\circ W'\circ D_{f'}).$$
\end{mycor}
\begin{proof}
Since $W$ is Shannon-equivalent to $D_g\circ W\circ D_f$ and $W'$ is Shannon-equivalent to $D_{g'}\circ W'\circ D_{f'}$, then $W$ is Shannon-equivalent to $W'$ if and only if $D_g\circ W\circ D_f$ is Shannon-equivalent to $D_{g'}\circ W'\circ D_{f'}$.
\end{proof}

\vspace*{3mm}

For every $W\in\DMC_{\mathcal{X}_1,\mathcal{Y}_1}$, we denote the $R_{\mathcal{X}_1,\mathcal{Y}_1}^{(s)}$-equivalence class of $W$ as $\hat{W}$, and for every $W\in\DMC_{\mathcal{X}_2,\mathcal{Y}_2}$, we denote the $R_{\mathcal{X}_2,\mathcal{Y}_2}^{(s)}$-equivalence class of $W$ as $\tilde{W}$.

\begin{myprop}
\label{propEmbedShanEquiv}
Let $\mathcal{X}_1,\mathcal{X}_2, \mathcal{Y}_1$ and $\mathcal{Y}_2$ be four finite sets such that $|\mathcal{X}_1|\leq |\mathcal{X}_2|$ and $|\mathcal{Y}_1|\leq |\mathcal{Y}_2|$. Let $f:\mathcal{X}_2\rightarrow\mathcal{X}_1$ be any fixed surjection from $\mathcal{X}_2$ to $\mathcal{X}_1$, and let $g:\mathcal{Y}_1\rightarrow\mathcal{Y}_2$ be any fixed injection from $\mathcal{Y}_1$ to $\mathcal{Y}_2$. Define the mapping $F:\DMC_{\mathcal{X}_1,\mathcal{Y}_1}^{(s)}\rightarrow \DMC_{\mathcal{X}_2,\mathcal{Y}_2}^{(s)}$ as
$F(\hat{W})=\widetilde{D_g\circ W'\circ D_f}=\Proj_2(D_g\circ W'\circ D_f)$, where $W'\in \hat{W}$, $\widetilde{D_g\circ W'\circ D_f}$ is the $R_{\mathcal{X}_2,\mathcal{Y}_2}^{(s)}$-equivalence class of $D_g\circ W'\circ D_f$, and $\Proj_2$ is the projection onto the $R_{\mathcal{X}_2,\mathcal{Y}_2}^{(s)}$-equivalence classes. We have:
\begin{itemize}
\item $F$ is well defined, i.e., $F(\hat{W})$ does not depend on $W'\in\hat{W}$.
\item $F$ is a homeomorphism between $\DMC_{\mathcal{X}_1,\mathcal{Y}_1}^{(s)}$ and $F\big(\DMC_{\mathcal{X}_1,\mathcal{Y}_1}^{(s)}\big)\subset \DMC_{\mathcal{X}_2,\mathcal{Y}_2}^{(s)}$.
\item $F$ does not depend on the surjection $f$ nor on the injection $g$. It depends only on $\mathcal{X}_1$, $\mathcal{X}_2$, $\mathcal{Y}_1$ and $\mathcal{Y}_2$, hence it is canonical.
\item $F\big(\DMC_{\mathcal{X}_1,\mathcal{Y}_1}^{(s)}\big)$ depends only on $|\mathcal{X}_1|$, $|\mathcal{Y}_1|$, $\mathcal{X}_2$ and $\mathcal{Y}_2$.
\item For every $W'\in\hat{W}$ and every $W''\in F(\hat{W})$, $W'$ is Shannon-equivalent to $W''$.
\end{itemize}
\end{myprop}
\begin{proof}
See Appendix \ref{appEmbedShanEquiv}.
\end{proof}

\begin{mycor}
\label{corIdentShanEquiv}
If $|\mathcal{X}_1|=|\mathcal{X}_2|$ and $|\mathcal{Y}_1|=|\mathcal{Y}_2|$, there exists a canonical homeomorphism from $\DMC_{\mathcal{X}_1,\mathcal{Y}_1}^{(s)}$ to $\DMC_{\mathcal{X}_2,\mathcal{Y}_2}^{(s)}$ depending only on $\mathcal{X}_1,\mathcal{Y}_1,\mathcal{X}_2$ and $\mathcal{Y}_2$.
\end{mycor}
\begin{proof}
Let $f$ be a bijection from $\mathcal{X}_2$ to $\mathcal{X}_1$, and let $g$ be a bijection from $\mathcal{Y}_1$ to $\mathcal{Y}_2$. Define the mapping $F:\DMC_{\mathcal{X}_1,\mathcal{Y}_1}^{(s)}\rightarrow \DMC_{\mathcal{X}_2,\mathcal{Y}_2}^{(s)}$ as
$F(\hat{W})=\widetilde{D_g\circ W'\circ D_f}=\Proj_2(D_g\circ W'\circ D_f),$ where $W'\in \hat{W}$ and $\Proj_2:\DMC_{\mathcal{X}_2,\mathcal{Y}_2}\rightarrow\DMC_{\mathcal{X}_2,\mathcal{Y}_2}^{(s)}$ is the projection onto the $R_{\mathcal{X}_2,\mathcal{Y}_2}^{(s)}$-equivalence classes.

Also, define the mapping $F':\DMC_{\mathcal{X}_2,\mathcal{Y}_2}^{(s)}\rightarrow \DMC_{\mathcal{X}_1,\mathcal{Y}_1}^{(s)}$ as
$$\textstyle F'(\tilde{V})=\widehat{D_{g^{-1}}\circ V'\circ D_{f^{-1}}}=\Proj_1(D_{g^{-1}}\circ V'\circ D_{f^{-1}}),$$ where $V'\in \tilde{V}$ and $\Proj_1:\DMC_{\mathcal{X}_1,\mathcal{Y}_1}\rightarrow\DMC_{\mathcal{X}_1,\mathcal{Y}_1}^{(s)}$ is the projection onto the $R_{\mathcal{X}_1,\mathcal{Y}_1}^{(s)}$-equivalence classes.

Proposition \ref{propEmbedShanEquiv} shows that $F$ and $F'$ are well defined.

For every $W\in\DMC_{\mathcal{X}_1,\mathcal{Y}_1}$, we have:
\begin{align*}
F'(F(\hat{W}))&\stackrel{(a)}{=}F'(\widetilde{D_g\circ W\circ D_f})\stackrel{(b)}{=}\textstyle\Proj_1(D_{g^{-1}}\circ(D_g\circ W\circ D_f)\circ D_{f^{-1}})=\hat{W},
\end{align*}
where (a) follows from the fact that $W\in \hat{W}$ and (b) follows from the fact that $D_g\circ W\circ D_f\in \widetilde{D_g\circ W\circ D_f}$.

We can similarly show that $F(F'(\tilde{V}))=\tilde{V}$ for every $\tilde{V}\in\DMC_{\mathcal{X}_2,\mathcal{Y}_2}^{(s)}$. Therefore, both $F$ and $F'$ are bijections. Proposition \ref{propEmbedShanEquiv} now implies that $F$ is a homeomorphism from $\DMC_{\mathcal{X}_1,\mathcal{Y}_1}^{(s)}$ to $F\big(\DMC_{\mathcal{X}_1,\mathcal{Y}_1}^{(s)}\big)=\DMC_{\mathcal{X}_2,\mathcal{Y}_2}^{(s)}$. Moreover, $F$ depends only on $\mathcal{X}_1,\mathcal{Y}_1,\mathcal{X}_2$ and $\mathcal{Y}_2$.
\end{proof}

\vspace*{3mm}

Corollary \ref{corIdentShanEquiv} allows us to identify $\DMC_{\mathcal{X},\mathcal{Y}}^{(s)}$ with $\DMC_{[n],[m]}^{(s)}$ through the canonical homeomorphism, where $n=|\mathcal{X}|$, $m=|\mathcal{Y}|$, $[n]=\{1,\ldots,n\}$ and $[m]=\{1,\ldots,m\}$. Moreover, for every $1\leq n\leq n'$ and $1\leq m\leq m'$, Proposition \ref{propEmbedShanEquiv} allows us to identify $\DMC_{[n],[m]}^{(s)}$ with the canonical subspace of $\DMC_{[n'],[m']}^{(s)}$ that is homeomorphic to $\DMC_{[n],[m]}^{(s)}$. In the rest of this paper, we consider that $\DMC_{[n],[m]}^{(s)}$ is a compact subspace of $\DMC_{[n'],[m']}^{(s)}$.

\begin{myconj}
\label{conjShanInterior}
For every $1\leq n<m$, the interior of $\DMC_{[n],[n]}^{(s)}$ in $\DMC_{[m],[m]}^{(s)}$ is empty.
\end{myconj}

\section{Space of Shannon-equivalent channels}

\label{secShanEquivSpace}

The previous section showed that if we are interested in Shannon-equivalent channels, it is sufficient to study the spaces $\DMC_{[n],[m]}$ and $\DMC_{[n],[m]}^{(s)}$ for every $n,m\geq 1$. Define the space $$\textstyle\DMC_{\ast,\ast}={\displaystyle\coprod_{\substack{n\geq 1,\\m\geq 1}}} \DMC_{[n],[m]}.$$
The subscripts $\ast$ indicate that the input and output alphabets of the considered channels are arbitrary but finite. We define the equivalence relation $R_{\ast,\ast}^{(s)}$ on $\DMC_{\ast,\ast}$ as follows:
$$WR_{\ast,\ast}^{(s)}W'\;\;\Leftrightarrow\;\;W\;\text{is Shannon-equivalent to}\;W'.$$

\begin{mydef}
The space of Shannon-equivalent channels is the quotient of the space of channels by the Shannon-equivalence relation:
$$\textstyle\DMC_{\ast,\ast}^{(s)}=\DMC_{\ast,\ast}/R_{\ast,\ast}^{(s)}.$$
\end{mydef}

Clearly, $\DMC_{[n],[m]}/R_{\ast,\ast}^{(s)}$ can be canonically identified with $\DMC_{[n],[m]}/R_{[n],[m]}^{(s)}=\DMC_{[n],[m]}^{(s)}$ for every $n,m\geq 1$. Therefore, we can write
\begin{align*}
\textstyle\DMC_{\ast,\ast}^{(s)}={\displaystyle\bigcup_{n,m\geq 1}}\DMC_{[n],[m]}^{(s)}\stackrel{(a)}{=}{\displaystyle\bigcup_{n\geq 1}}\DMC_{[n],[n]}^{(s)}.
\end{align*}
Note that (a) follows from the fact that $\DMC_{[n],[m]}^{(s)}\subset \DMC_{[k],[k]}^{(s)}$ (see Section \ref{subsecEmbedShanEquiv}), where $k=\max\{n,m\}$.

We define the \emph{Shannon-rank} of $\hat{W}\in\DMC_{\ast,\ast}^{(s)}$ as: $$\srank(\hat{W})=\min\{n\geq 1:\; \hat{W}\in{\DMC}_{[n],[n]}^{(s)}\}.$$ Clearly,
$$\textstyle{\DMC}_{[n],[n]}^{(s)}=\{\hat{W}\in{\DMC}_{\ast,\ast}^{(s)}:\;\srank(\hat{W})\leq n\}.$$

A subset $A$ of $\DMC_{\ast,\ast}^{(s)}$ is said to be \emph{rank-bounded} if there exists $n\geq 1$ such that $A\subset \DMC_{[n],[n]}^{(s)}$.

\subsection{Natural topologies on $\DMC_{\ast,\ast}^{(s)}$}

\label{subsecNaturalTopShan}

Since $\DMC_{\ast,\ast}^{(s)}$ is the quotient of $\DMC_{\ast,\ast}$ and since $\DMC_{\ast,\ast}$ was not given any topology, there is no ``standard topology" on $\DMC_{\ast,\ast}^{(s)}$. However, there are many properties that one may require from any ``reasonable" topology on $\DMC_{\ast,\ast}^{(s)}$. In this paper, we focus on one particular requirement that we consider the most basic property required from any ``acceptable" topology on $\DMC_{\ast,\ast}^{(s)}$:

\begin{mydef}
A topology $\mathcal{T}$ on $\DMC_{\ast,\ast}^{(s)}$ is said to be \emph{natural} if it induces the quotient topology $\mathcal{T}_{[n],[m]}^{(s)}$ on $\DMC_{[n],[m]}^{(s)}$ for every $n,m\geq 1$.
\end{mydef}

The reason why we consider such topology as natural is because the quotient topology $\mathcal{T}_{[n],[m]}^{(s)}$ is the ``standard" and ``most natural" topology on $\DMC_{[n],[m]}^{(s)}$. Therefore, we do not want to induce any non-standard topology on $\DMC_{[n],[m]}^{(s)}$ by relativization.

\begin{myprop}
\label{propShanNaturalTopProperties}
Every natural topology is $\sigma$-compact, separable and path-connected.
\end{myprop}
\begin{proof}
Since $\DMC_{\ast,\ast}^{(s)}$ is the countable union of compact and separable subspaces (namely $\{\DMC_{[n],[n]}^{(s)}\}_{n\geq 1}$), $\DMC_{\ast,\ast}^{(s)}$ is $\sigma$-compact and separable.

On the other hand, since $\displaystyle\bigcap_{n\geq 1}\textstyle\DMC_{[n],[n]}^{(s)}=\DMC_{[1],[1]}^{(s)}\neq\o$ and since $\DMC_{[n],[n]}^{(s)}$ is path-connected for every $n\geq1$, the union $\textstyle\DMC_{\ast,\ast}^{(s)}={\displaystyle\bigcup_{n\geq 1}}\DMC_{[n],[n]}^{(s)}$ is path-connected.
\end{proof}

\begin{myrem}
\label{remNegPropertiesNatTopShan}
It is possible to show that if Conjecture \ref{conjShanInterior} is true, then for every natural topology $\mathcal{T}$ on $\DMC_{\ast,\ast}^{(s)}$, we have:
\begin{itemize}
\item Every open set is rank-unbounded.
\item For every $n\geq 1$, the interior of $\DMC_{[n],[n]}^{(s)}$ in $(\DMC_{\ast,\ast}^{(s)},\mathcal{T})$ is empty.
\item If $\mathcal{T}$ is Hausdorff, then
\begin{itemize}
\item $(\DMC_{\ast,\ast}^{(s)},\mathcal{T})$ is not a Baire space, hence no natural topology can be completely metrized.
\item $(\DMC_{\ast,\ast}^{(s)},\mathcal{T})$ is not locally compact anywhere.
\end{itemize}
\end{itemize}
\end{myrem}

\section{Strong topology on $\DMC_{\ast,\ast}^{(s)}$}

\label{subsecInStrongTop}

Since the spaces $\{\DMC_{[n],[m]}\}_{n,m\geq 1}$ are disjoint and since there is no a priori way to (topologically) compare channels in $\DMC_{[n],[m]}$ with channels in $\DMC_{[n'],[m']}$ for $(n,m)\neq(n',m')$, the ``most natural" topology that we can define on $\DMC_{\ast,\ast}$ is the disjoint union topology $\mathcal{T}_{s,\ast,\ast}:=\displaystyle\bigoplus_{n,m\geq 1}\mathcal{T}_{[n],[m]}$. Clearly, the space $(\DMC_{\ast,\ast},\mathcal{T}_{s,\ast,\ast})$ is disconnected. Moreover, $\mathcal{T}_{s,\ast,\ast}$ is metrizable because it is the disjoint union of metrizable spaces. It is also $\sigma$-compact because it is the union of countably many compact spaces.

We added the subscript $s$ to emphasize the fact that $\mathcal{T}_{s,\ast,\ast}$ is a strong topology (remember that the disjoint union topology is the \emph{finest} topology that makes the canonical injections continuous).

\begin{mydef}
We define the strong topology $\mathcal{T}_{s,\ast,\ast}^{(s)}$ on $\DMC_{\ast,\ast}^{(s)}$ as the quotient topology $\mathcal{T}_{s,\ast,\ast}/R_{\ast,\ast}^{(s)}$.

We call open and closed sets in $(\DMC_{\ast,\ast}^{(s)},\mathcal{T}_{s,\ast,\ast}^{(s)})$ as strongly open and strongly closed sets respectively.
\end{mydef}

Let $\Proj:\DMC_{\ast,\ast}\rightarrow\DMC_{\ast,\ast}^{(s)}$ be the projection onto the $R_{\ast,\ast}^{(s)}$-equivalence classes, and for every $n,m\geq 1$ let $\Proj_{n,m}:\DMC_{[n],[m]}\rightarrow\DMC_{[n],[m]}^{(s)}$ be the projection onto the $R_{[n],[m]}^{(s)}$-equivalence classes. Due to the identifications that we made in Section \ref{secShanEquivSpace}, we have $\Proj(W)=\Proj_{n,m}(W)$ for every $W\in\DMC_{[n],[m]}$. Therefore, for every $U\subset \DMC_{\ast,\ast}^{(s)}$, we have
$$\textstyle\Proj^{-1}(U)={\displaystyle\coprod_{n,m\geq 1}}\Proj_{n,m}^{-1}(U\cap\DMC_{[n],[m]}^{(s)}).$$
Hence,
\begin{align*}
\textstyle U\in\mathcal{T}_{s,\ast,\ast}^{(s)}\;\;&\stackrel{(a)}{\Leftrightarrow}\;\;\textstyle\Proj^{-1}(U)\in\mathcal{T}_{s,\ast,\ast}\\
&\textstyle\stackrel{(b)}{\Leftrightarrow}\;\;\Proj^{-1}(U)\cap\DMC_{[n],[m]} \in\mathcal{T}_{[n],[m]},\;\;\forall n,m\geq 1\\
&\textstyle\Leftrightarrow\;\;\left({\displaystyle\coprod_{n',m'\geq 1}}\Proj_{n',m'}^{-1}(U\cap\DMC_{[n'],[m']}^{(s)})\right)\cap\DMC_{[n],[m]} \in\mathcal{T}_{[n],[m]},\;\;\forall n,m\geq 1\\
&\textstyle\Leftrightarrow\;\;\Proj_{n,m}^{-1}(U\cap\DMC_{[n],[m]}^{(s)}) \in\mathcal{T}_{[n],[m]},\;\;\forall n,m\geq 1\\
&\textstyle\stackrel{(c)}{\Leftrightarrow}\;\;U\cap\DMC_{[n],[m]}^{(s)} \in\mathcal{T}_{[n],[m]}^{(s)},\;\;\forall n,m\geq 1,
\end{align*}
where (a) and (c) follow from the properties of the quotient topology, and (b) follows from the properties of the disjoint union topology.

We conclude that $U\subset \DMC_{\ast,\ast}^{(s)}$ is strongly open in $\DMC_{\ast,\ast}^{(s)}$ if and only if $U\cap \DMC_{[n],[m]}^{(s)}$ is open in $\DMC_{[n],[m]}^{(s)}$ for every $n,m\geq 1$. This shows that the topology on $\DMC_{[n],[m]}^{(s)}$ that is inherited from $(\DMC_{\ast,\ast}^{(s)},\mathcal{T}_{s,\ast,\ast}^{(s)})$ is exactly $\mathcal{T}_{[n],[m]}^{(s)}$. Therefore, $\mathcal{T}_{s,\ast,\ast}^{(s)}$ is a natural topology. On the other hand, if $\mathcal{T}$ is an arbitrary natural topology and $U\in\mathcal{T}$, then $U\cap\DMC_{[n],[m]}^{(s)}$ is open in $\DMC_{[n],[m]}^{(s)}$ for every $n,m\geq 1$, so $U\in\mathcal{T}_{s,\ast,\ast}^{(s)}$. We conclude that $\mathcal{T}_{s,\ast,\ast}^{(s)}$ is the finest natural topology.

\vspace*{3mm}

We can also characterize the strongly closed subsets of $\DMC_{\ast,\ast}^{(s)}$ in terms of the closed sets of the $\DMC_{[n],[m]}^{(s)}$ spaces:
\begin{align*}
F\;\text{is strongly closed in}\;&\textstyle\DMC_{\ast,\ast}^{(s)}\;\;\\
&\Leftrightarrow\;\;\textstyle\DMC_{\ast,\ast}^{(s)}\setminus F\;\text{is strongly open in}\;\textstyle\DMC_{\ast,\ast}^{(s)}\\
&\textstyle\Leftrightarrow\;\;\left(\DMC_{\ast,\ast}^{(s)}\setminus F\right)\cap\DMC_{[n],[m]}^{(s)}\;\text{is open in}\;\DMC_{[n],[m]}^{(s)},\;\;\forall n,m\geq 1\\
&\textstyle\Leftrightarrow\;\;\DMC_{[n],[m]}^{(s)}\setminus \left(F\cap\DMC_{[n],[m]}^{(s)}\right)\;\text{is open in}\;\DMC_{[n],[m]}^{(s)},\;\;\forall n,m\geq 1\\
&\textstyle\Leftrightarrow\;\;F\cap\DMC_{[n],[m]}^{(s)}\;\text{is closed in}\;\DMC_{[n],[m]}^{(s)},\;\;\forall n,m\geq 1.
\end{align*}

\begin{mylem}
For every subset $U$ of $\DMC_{\ast,\ast}^{(s)}$, we have:
\begin{itemize}
\item $U$ is strongly open if and only if $U\cap\DMC_{[n],[n]}^{(s)}$ is open in $\DMC_{[n],[n]}^{(s)}$ for every $n\geq 1$.
\item $U$ is strongly closed if and only if $U\cap\DMC_{[n],[n]}^{(s)}$ is closed in $\DMC_{[n],[n]}^{(s)}$ for every $n\geq 1$.
\end{itemize}
\end{mylem}
\begin{proof}
If $U$ is strongly open then $U\cap\DMC_{[n],[m]}^{(s)}$ is open in $\DMC_{[n],[m]}^{(s)}$ for every $n,m\geq 1$. This implies that $U\cap\DMC_{[n],[n]}^{(s)}$ is open in $\DMC_{[n],[n]}^{(s)}$ for every $n\geq 1$.

Conversely, assume that $U\cap\DMC_{[n],[n]}^{(s)}$ is open in $\DMC_{[n],[n]}^{(s)}$ for every $n\geq 1$. Fix $n,m\geq 1$ and let $k=\max\{n,m\}$. We have $\DMC_{[n],[m]}^{(s)}\subset\DMC_{[k],[k]}^{(s)}$. Since $U\cap \DMC_{[k],[k]}^{(s)}$ is open in $\DMC_{[k],[k]}^{(s)}$, the set $U\cap\DMC_{[n],[m]}^{(s)}=(U\cap \DMC_{[k],[k]}^{(s)})\cap \DMC_{[n],[m]}^{(s)}$ is open in $\DMC_{[n],[m]}^{(s)}$. Therefore, $U\cap\DMC_{[n],[m]}^{(s)}$ is open in $\DMC_{[n],[m]}^{(s)}$ for every $n,m\geq 1$, which implies that $U$ is strongly open.

We can similarly show that $U$ is strongly closed if and only if $U\cap\DMC_{[n],[n]}^{(s)}$ is closed in $\DMC_{[n],[n]}^{(s)}$ for every $n\geq 1$.
\end{proof}

\vspace*{3mm}

Since $\DMC_{[n],[n]}^{(s)}$ is metrizable for every $n\geq 1$, it is also normal. We can use this fact to prove that the strong topology on $\DMC_{\ast,\ast}^{(s)}$ is normal:

\begin{mylem}
\label{lemDMCXsNorm}
$(\DMC_{\ast,\ast}^{(s)},\mathcal{T}_{s,\ast,\ast}^{(s)})$ is normal.
\end{mylem}
\begin{proof}
See Appendix \ref{appDMCXsNorm}.
\end{proof}

\vspace*{3mm}

The following theorem shows that the strong topology satisfies many desirable properties.

\begin{mythe}
\label{theDMCXs}
$(\DMC_{\ast,\ast}^{(s)},\mathcal{T}_{s,\ast,\ast}^{(s)})$ is a compactly generated, sequential and $T_4$ space.
\end{mythe}
\begin{proof}
Since $(\DMC_{\ast,\ast},\mathcal{T}_{s,\ast,\ast})$ is metrizable, it is sequential. Therefore, $(\DMC_{\ast,\ast}^{(s)},\mathcal{T}_{s,\ast,\ast}^{(s)})$, which is the quotient of a sequential space, is sequential.

Let us now show that $\DMC_{\ast,\ast}^{(s)}$ is $T_4$. Fix $\hat{W}\in\DMC_{\ast,\ast}^{(s)}$. For every $n\geq 1$, we have $\{\hat{W}\}\cap \DMC_{[n],[n]}^{(s)}$ is either $\o$ or $\{\hat{W}\}$ depending on whether $\hat{W}\in \DMC_{[n],[n]}^{(s)}$ or not. Since $\DMC_{[n],[n]}^{(s)}$ is metrizable, it is $T_1$ and so singletons are closed in $\DMC_{[n],[n]}^{(s)}$. We conclude that in all cases, $\{\hat{W}\}\cap \DMC_{[n],[n]}^{(s)}$ is closed in $\DMC_{[n],[n]}^{(s)}$ for every $n\geq 1$. Therefore, $\{\hat{W}\}$ is strongly closed in $\DMC_{\ast,\ast}^{(s)}$. This shows that $(\DMC_{\ast,\ast}^{(s)},\mathcal{T}_{s,\ast,\ast}^{(s)})$ is $T_1$. On the other hand, Lemma \ref{lemDMCXsNorm} shows that $(\DMC_{\ast,\ast}^{(s)},\mathcal{T}_{s,\ast,\ast}^{(s)})$ is normal. This means that $(\DMC_{\ast,\ast}^{(s)},\mathcal{T}_{s,\ast,\ast}^{(s)})$ is $T_4$, which implies that it is Hausdorff.

Now since $(\DMC_{\ast,\ast},\mathcal{T}_{s,\ast,\ast})$ is metrizable, it is compactly generated. On the other hand, the quotient space $(\DMC_{\ast,\ast}^{(s)},\mathcal{T}_{s,\ast,\ast}^{(s)})$ was shown to be Hausdorff. We conclude that $(\DMC_{\ast,\ast}^{(s)},\mathcal{T}_{s,\ast,\ast}^{(s)})$ is compactly generated.
\end{proof}

\begin{myrem}
It is possible to show that if Conjecture \ref{conjShanInterior} is true, then we have:
\begin{itemize}
\item $\mathcal{T}_{s,\ast,\ast}^{(s)}$ is not first-countable anywhere.
\item A subset of $\DMC_{\ast,\ast}^{(s)}$ is compact in $\mathcal{T}_{s,\ast,\ast}$ if and only if it is rank-bounded and strongly closed.
\end{itemize}
\end{myrem}

\section{The BRM metric on the space of Shannon-equivalent channels}

We define the \emph{BRM metric} on $\DMC_{\ast,\ast}^{(s)}$ as follows:
$$d_{\ast,\ast}^{(s)}(\hat{W}_1,\hat{W}_2)=\sup_{\substack{n,m\geq 1,\\l\in{\Delta}_{[n]\times[m]}}}|\$_{\opt}(l,\hat{W}_1)-\$_{\opt}(l,\hat{W}_2)|.$$

Let $\mathcal{T}_{\ast,\ast}^{(s)}$ be the metric topology on $\DMC_{\ast,\ast}^{(s)}$ that is induced by $d_{\ast,\ast}^{(s)}$.  We call $\mathcal{T}_{\ast,\ast}^{(s)}$ the \emph{BRM topology} on $\DMC_{\ast,\ast}^{(s)}$. 

Clearly, $\mathcal{T}_{\ast,\ast}^{(s)}$ is natural because the restriction of $d_{\ast,\ast}^{(s)}$ on $\DMC_{[n],[m]}^{(s)}$ is exactly $d_{[n],[m]}^{(s)}$, and the topology induced by $d_{[n],[m]}^{(s)}$ is $\mathcal{T}_{[n],[m]}^{(s)}$ (Theorem \ref{theDMCXYs}).

\section{Continuity of channel parameters and operations in the strong topology}

\subsection{Channel parameters}

For every $W\in\DMC_{\ast,\ast}$, $C(W)$ depends only on the Shannon-equivalence class of $W$ \cite{ShannonDegrad}. Therefore, for every $\hat{W}\in\DMC_{\ast,\ast}^{(s)}$, we can define $C(\hat{W}):=C(W')$ for any $W'\in\hat{W}$. We can define $P_{e,n,M}(\hat{W})$ similarly.

\begin{myprop}
\label{propContParamDMCXYsStr}
Let $\mathcal{X}$ and $\mathcal{Y}$ be two finite sets. We have:
\begin{itemize}
\item $C:\DMC_{\mathcal{X},\mathcal{Y}}^{(s)}\rightarrow \mathbb{R}^+$ is continuous on $(\DMC_{\mathcal{X},\mathcal{Y}}^{(s)},\mathcal{T}_{\mathcal{X},\mathcal{Y}}^{(s)})$.
\item For every $n\geq 1$ and every $M\geq 1$, the mapping $P_{e,n,M}:\DMC_{\mathcal{X},\mathcal{Y}}^{(s)}\rightarrow [0,1]$ is continuous on $(\DMC_{\mathcal{X},\mathcal{Y}}^{(s)},\mathcal{T}_{\mathcal{X},\mathcal{Y}}^{(s)})$.
\end{itemize}
\end{myprop}
\begin{proof}
Since $C:\DMC_{\mathcal{X},\mathcal{Y}}\rightarrow\mathbb{R}^+$ is continuous, and since $C(W)$ depends only on the $R_{\mathcal{X},\mathcal{Y}}^{(s)}$-equivalence class of $W$, Lemma \ref{lemQuotientFunction} implies that $C:\DMC_{\mathcal{X},\mathcal{Y}}^{(s)}\rightarrow \mathbb{R}^+$ is continuous on $(\DMC_{\mathcal{X},\mathcal{Y}}^{(s)},\mathcal{T}_{\mathcal{X},\mathcal{Y}}^{(s)})$. We can show the continuity of $P_{e,n,M}$ on $(\DMC_{\mathcal{X},\mathcal{Y}}^{(s)},\mathcal{T}_{\mathcal{X},\mathcal{Y}}^{(s)})$ similarly.
\end{proof}

\vspace*{3mm}

The following lemma provides a way to check whether a mapping defined on $(\DMC_{\ast,\ast}^{(s)},\mathcal{T}_{s,\ast,\ast}^{(s)})$ is continuous:

\begin{mylem}
\label{lemContinuityForStrongTopShan}
Let $(S,\mathcal{V})$ be an arbitrary topological space. A mapping $f:\DMC_{\ast,\ast}^{(s)}\rightarrow S$ is continuous on $(\DMC_{\ast,\ast}^{(s)},\mathcal{T}_{s,\ast,\ast}^{(s)})$ if and only if it is continuous on $(\DMC_{[n],[n]}^{(s)},\mathcal{T}_{[n],[n]}^{(s)})$ for every $n\geq 1$.
\end{mylem}
\begin{proof}
\begin{align*}
\textstyle f\;\text{is continuous on}\;(\DMC_{\ast,\ast}^{(s)},\mathcal{T}_{s,\ast,\ast}^{(s)})\;\;&\textstyle\Leftrightarrow\;\; f^{-1}(V)\in \mathcal{T}_{s,\ast,\ast}^{(s)},\;\;\forall V\in\mathcal{V}\\
&\textstyle\Leftrightarrow\;\; f^{-1}(V)\cap \DMC_{[n],[n]}^{(s)} \in \mathcal{T}_{[n],[n]}^{(s)},\;\;\forall n\geq 1,\; \forall V\in\mathcal{V}\\
&\textstyle\Leftrightarrow\;\; f\;\text{is continuous on}\; (\DMC_{[n],[n]}^{(s)},\mathcal{T}_{[n],[n]}^{(s)}),\;\;\forall n\geq 1.
\end{align*}
\end{proof}

\begin{myprop}
\label{propContParamDMCsStr}
We have:
\begin{itemize}
\item $C:\DMC_{\ast,\ast}^{(s)}\rightarrow \mathbb{R}^+$ is continuous on $(\DMC_{\ast,\ast}^{(s)},\mathcal{T}_{s,\ast,\ast}^{(s)})$.
\item For every $n\geq 1$ and every $M\geq 1$, the mapping $P_{e,n,M}:\DMC_{\ast,\ast}^{(s)}\rightarrow [0,1]$ is continuous on $(\DMC_{\ast,\ast}^{(s)},\mathcal{T}_{s,\ast,\ast}^{(s)})$.
\end{itemize}
\end{myprop}
\begin{proof}
The proposition follows from Proposition \ref{propContParamDMCXYsStr} and Lemma \ref{lemContinuityForStrongTopShan}.
\end{proof}

\subsection{Channel operations}

\label{subsecContOperShan}

Channel sums and products can be ``quotiented" by the Shannon-equivalence relation. We just need to realize that the Shannon-equivalence class of the resulting channel depends only on the Shannon-equivalence classes of the channels that were used in the operation \cite{ShannonDegrad}.

\begin{myprop}
\label{propContOperDMCXYs}
We have:
\begin{itemize}
\item The mapping $(\hat{W}_1,\overline{W}_2)\rightarrow \hat{W}_1\oplus \overline{W}_2$ from $\DMC_{\mathcal{X}_1,\mathcal{Y}_1}^{(s)}\times \DMC_{\mathcal{X}_2,\mathcal{Y}_2}^{(s)}$ to $\DMC_{\mathcal{X}_1\coprod \mathcal{X}_2,\mathcal{Y}_1\coprod\mathcal{Y}_2}^{(s)}$ is continuous.
\item The mapping $(\hat{W}_1,\overline{W}_2)\rightarrow \hat{W}_1\otimes \overline{W}_2$ from $\DMC_{\mathcal{X}_1,\mathcal{Y}_1}^{(s)}\times \DMC_{\mathcal{X}_2,\mathcal{Y}_2}^{(s)}$ to $\DMC_{\mathcal{X}_1\times\mathcal{X}_2,\mathcal{Y}_1\times\mathcal{Y}_2}^{(s)}$ is continuous.
\end{itemize}
\end{myprop}
\begin{proof}
We only prove the continuity of the channel sum because the proof for the channel product is similar.

Let $\Proj:\DMC_{\mathcal{X}_1\coprod\mathcal{X}_2,\mathcal{Y}_1\coprod\mathcal{Y}_2}\rightarrow \DMC_{\mathcal{X}_1\coprod\mathcal{X}_2,\mathcal{Y}_1\coprod\mathcal{Y}_2}^{(s)}$ be the projection onto the $R_{\mathcal{X}_1\coprod\mathcal{X}_2,\mathcal{Y}_1\coprod\mathcal{Y}_2}^{(s)}$-equivalence classes. Define the mapping $f:\DMC_{\mathcal{X}_1,\mathcal{Y}_1}\times \DMC_{\mathcal{X}_2,\mathcal{Y}_2}\rightarrow \DMC_{\mathcal{X}_1\coprod\mathcal{X}_2,\mathcal{Y}_1\coprod\mathcal{Y}_2}^{(s)}$ as $f(W_1,W_2)=\Proj(W_1\oplus W_2)$. Clearly, $f$ is continuous.

Now define the equivalence relation $R$ on $\DMC_{\mathcal{X}_1,\mathcal{Y}_1}\times \DMC_{\mathcal{X}_2,\mathcal{Y}_2}$ as:
$$(W_1,W_2)R(W_1',W_2')\;\;\Leftrightarrow\;\; W_1 R_{\mathcal{X}_1,\mathcal{Y}_1}^{(s)}W_1'\;\text{and}\;W_2 R_{\mathcal{X}_2,\mathcal{Y}_2}^{(s)}W_2'.$$
The discussion before the proposition shows that $f(W_1,W_2)=\Proj(W_1\oplus W_2)$ depends only on the $R$-equivalence class of $(W_1,W_2)$. Lemma \ref{lemQuotientFunction} now shows that the transcendent map of $f$ defined on $(\DMC_{\mathcal{X}_1,\mathcal{Y}_1}\times \DMC_{\mathcal{X}_2,\mathcal{Y}_2})/R$ is continuous.

Notice that $(\DMC_{\mathcal{X}_1,\mathcal{Y}_1}\times \DMC_{\mathcal{X}_2,\mathcal{Y}_2})/R$ can be identified with $\DMC_{\mathcal{X}_1,\mathcal{Y}_1}^{(s)}\times \DMC_{\mathcal{X}_2,\mathcal{Y}_2}^{(s)}$. Therefore, we can define $f$ on $\DMC_{\mathcal{X}_1,\mathcal{Y}_1}^{(s)}\times \DMC_{\mathcal{X}_2,\mathcal{Y}_2}^{(s)}$ through this identification. Moreover, since $\DMC_{\mathcal{X}_1,\mathcal{Y}_1}$ and $\DMC_{\mathcal{X}_2,\mathcal{Y}_2}^{(s)}$ are locally compact and Hausdorff, Corollary \ref{corQuotientProd} implies that the canonical bijection between $(\DMC_{\mathcal{X}_1,\mathcal{Y}_1}\times \DMC_{\mathcal{X}_2,\mathcal{Y}_2})/R$ and $\DMC_{\mathcal{X}_1,\mathcal{Y}_1}^{(s)}\times \DMC_{\mathcal{X}_2,\mathcal{Y}_2}^{(s)}$ is a homeomorphism. 

Now since the mapping $f$ on $\DMC_{\mathcal{X}_1,\mathcal{Y}_1}^{(s)}\times \DMC_{\mathcal{X}_2,\mathcal{Y}_2}^{(s)}$ is just the channel sum, we conclude that the mapping $(\hat{W}_1,\overline{W}_2)\rightarrow \hat{W}_1\oplus \overline{W}_2$ from $\DMC_{\mathcal{X}_1,\mathcal{Y}_1}^{(s)}\times \DMC_{\mathcal{X}_2,\mathcal{Y}_2}^{(s)}$ to $\DMC_{\mathcal{X}_1\coprod\mathcal{X}_2,\mathcal{Y}_1\coprod\mathcal{Y}_2}^{(s)}$ is continuous.
\end{proof}

\begin{myprop}
\label{propContOperDMCsStr}
Assume that the space $\DMC_{\ast,\ast}^{(s)}$ is endowed with the strong topology. We have:
\begin{itemize}
\item The mapping $(\hat{W}_1,\overline{W}_2)\rightarrow \hat{W}_1\oplus \overline{W}_2$ from $\DMC_{\ast,\ast}^{(s)}\times \DMC_{\mathcal{X}_2,\mathcal{Y}_2}^{(s)}$ to $\DMC_{\ast,\ast}^{(s)}$ is continuous.
\item The mapping $(\hat{W}_1,\overline{W}_2)\rightarrow \hat{W}_1\otimes \overline{W}_2$ from $\DMC_{\ast,\ast}^{(s)}\times \DMC_{\mathcal{X}_2,\mathcal{Y}_2}^{(s)}$ to $\DMC_{\ast,\ast}^{(s)}$ is continuous.
\end{itemize}
\end{myprop}
\begin{proof}
We only prove the continuity of the channel sum because the proof of the continuity of the channel product is similar.

Due to the distributivity of the product with respect to disjoint unions, we have:
$$\textstyle\DMC_{\ast,\ast}\times\DMC_{\mathcal{X}_2,\mathcal{Y}_2}={\displaystyle\coprod_{n,m\geq1}}(\DMC_{[n],[m]}\times\DMC_{\mathcal{X}_2,\mathcal{Y}_2}),$$
and
$$\textstyle\mathcal{T}_{s,\ast,\ast}\otimes\mathcal{T}_{\mathcal{X}_2,\mathcal{Y}_2}={\displaystyle\bigoplus_{n,m\geq1}}\left(\mathcal{T}_{[n],[m]}\otimes\mathcal{T}_{\mathcal{X}_2,\mathcal{Y}_2}\right).$$

Therefore, the space $\DMC_{\ast,\ast}\times\DMC_{\mathcal{X}_2,\mathcal{Y}_2}$ is the topological disjoint union of the spaces $(\DMC_{[n],[m]}\times\DMC_{\mathcal{X}_2,\mathcal{Y}_2})_{n,m\geq 1}$.

For every $n,m\geq 1$, let $\Proj_{n,m}$ be the projection onto the $R_{[n]\coprod\mathcal{X}_2,[m]\coprod\mathcal{Y}_2}^{(s)}$-equivalence classes and let $i_{n,m}$ be the canonical injection from $\DMC_{[n]\coprod\mathcal{X}_2,[m]\coprod\mathcal{Y}_2}^{(s)}$ to $\DMC_{\ast,\ast}^{(s)}$. 

Define the mapping $f: \DMC_{\ast,\ast}\times\DMC_{\mathcal{X}_2,\mathcal{Y}_2}\rightarrow \DMC_{\ast,\ast}^{(s)}$ as $$\textstyle f(W_1,W_2)=i_{n,m}(\Proj_{n,m}(W_1\oplus W_2))=\hat{W}_1\oplus\overline{W}_2,$$
where $n$ and $m$ are the unique integers satisfying $W_1\in \DMC_{[n],[m]}$. $\hat{W}_1$ and $\overline{W}_2$ are the $R_{[n],[m]}^{(s)}$ and $R_{\mathcal{X}_2,\mathcal{Y}_2}^{(s)}$-equivalence classes of $W_1$ and $W_2$ respectively.

Clearly, the mapping $f$ is continuous on $\DMC_{[n],[m]}\times\DMC_{\mathcal{X}_2,\mathcal{Y}_2}$ for every $n,m\geq 1$. Therefore, $f$ is continuous on $(\DMC_{\ast,\ast}\times\DMC_{\mathcal{X}_2,\mathcal{Y}_2},\mathcal{T}_{s,\ast,\ast}\otimes\mathcal{T}_{\mathcal{X}_2,\mathcal{Y}_2})$.

Let $R$ be the equivalence relation defined on $\DMC_{\ast,\ast}\times\DMC_{\mathcal{X}_2,\mathcal{Y}_2}$ as follows: $(W_1,W_2)R(W_1',W_2')$ if and only if $W_1 R_{\ast,\ast}^{(s)} W_1'$ and $W_2 R_{\mathcal{X}_2,\mathcal{Y}_2}^{(s)} W_2'$.

Since $f(W_1,W_2)$ depends only on the $R$-equivalence class of $(W_1,W_2)$, Lemma \ref{lemQuotientFunction} implies that the transcendent mapping of $f$ is continuous on $(\DMC_{\ast,\ast}\times\DMC_{\mathcal{X}_2,\mathcal{Y}_2})/R$.

Since $(\DMC_{\ast,\ast},\mathcal{T}_{s,\ast,\ast})$ and $\DMC_{\mathcal{X}_2,\mathcal{Y}_2}^{(s)}=\DMC_{\mathcal{X}_2,\mathcal{Y}_2}/R_{\mathcal{X}_2,\mathcal{Y}_2}^{(s)}$ are Hausdorff and locally compact, Corollary \ref{corQuotientProd} implies that the canonical bijection from $\DMC_{\ast,\ast}^{(s)}\times \DMC_{\mathcal{X}_2,\mathcal{Y}_2}^{(s)}$ to $(\DMC_{\ast,\ast}\times \DMC_{\mathcal{X}_2,\mathcal{Y}_2})/R$ is a homeomorphism. We conclude that the channel sum is continuous on $(\DMC_{\ast,\ast}^{(s)}\times \DMC_{\mathcal{X}_2,\mathcal{Y}_2}^{(s)},\mathcal{T}_{s,\ast,\ast}^{(s)}\otimes\mathcal{T}_{\mathcal{X}_2,\mathcal{Y}_2}^{(s)})$.
\end{proof}

\vspace*{3mm}

The reader might be wondering why the channel sum and the channel product were not shown to be continuous on the whole space $\DMC_{\ast,\ast}^{(s)}\times \DMC_{\ast,\ast}^{(s)}$ instead of the smaller space $\DMC_{\ast,\ast}^{(s)}\times \DMC_{\mathcal{X}_2,\mathcal{Y}_2}^{(s)}$. The reason is because we cannot apply Corollary \ref{corQuotientProd} to $\DMC_{\ast,\ast}\times \DMC_{\ast,\ast}$ and $\DMC_{\ast,\ast}^{(s)}\times \DMC_{\ast,\ast}^{(s)}$ since we do not know whether $(\DMC_{\ast,\ast}^{(s)},\mathcal{T}_{s,\ast,\ast}^{(s)})$ is locally compact or not. Moreover, as we stated in Remark \ref{remNegPropertiesNatTopShan}, if Conjecture \ref{conjShanInterior} is true then $(\DMC_{\ast,\ast}^{(s)},\mathcal{T}_{s,\ast,\ast}^{(s)})$ is not locally compact.

As in the case of the space of equivalent channels \cite{RajContTop}, one potential method to show the continuity of the channel sum on $(\DMC_{\ast,\ast}^{(s)}\times\DMC_{\ast,\ast}^{(s)},\mathcal{T}_{s,\ast,\ast}^{(s)}\otimes \mathcal{T}_{s,\ast,\ast}^{(s)})$ is as follows: let $R$ be the equivalence relation on $\DMC_{\ast,\ast}\times\DMC_{\ast,\ast}$ defined as $(W_1,W_2)R(W_1',W_2')$ if and only if $W_1 R_{\ast,\ast}^{(s)}W_1'$ and $W_2 R_{\ast,\ast}^{(s)}W_2'$. We can identify $(\DMC_{\ast,\ast}\times\DMC_{\ast,\ast})/R$ with $\DMC_{\ast,\ast}^{(s)}\times\DMC_{\ast,\ast}^{(s)}$ through the canonical bijection. Using Lemma \ref{lemQuotientFunction}, it is easy to see that the mapping $(\hat{W}_1,\overline{W}_2)\rightarrow \hat{W}_1\oplus\overline{W}_2$ is continuous from $\big(\DMC_{\ast,\ast}^{(s)}\times\DMC_{\ast,\ast}^{(s)}, (\mathcal{T}_{s,\ast,\ast}\otimes \mathcal{T}_{s,\ast,\ast})/R\big)$ to $(\DMC_{\ast,\ast}^{(s)},\mathcal{T}_{s,\ast,\ast}^{(s)})$.

It was shown in \cite{CompactlyGenerated} that the topology $(\mathcal{T}_{s,\ast,\ast}\otimes \mathcal{T}_{s,\ast,\ast})/R$ is homeomorphic to $\kappa(\mathcal{T}_{s,\ast,\ast}^{(s)}\otimes \mathcal{T}_{s,\ast,\ast}^{(s)})$ through the canonical bijection, where $\kappa(\mathcal{T}_{s,\ast,\ast}^{(s)}\otimes \mathcal{T}_{s,\ast,\ast}^{(s)})$ is the coarsest topology that is both compactly generated and finer than $\mathcal{T}_{s,\ast,\ast}^{(s)}\otimes \mathcal{T}_{s,\ast,\ast}^{(s)}$. Therefore, the mapping $(\hat{W}_1,\overline{W}_2)\rightarrow \hat{W}_1\oplus\overline{W}_2$ is continuous on $\big(\DMC_{\ast,\ast}^{(s)}\times\DMC_{\ast,\ast}^{(s)}, \kappa(\mathcal{T}_{s,\ast,\ast}^{(s)}\otimes \mathcal{T}_{s,\ast,\ast}^{(s)})\big)$. This means that if $\mathcal{T}_{s,\ast,\ast}^{(s)}\otimes \mathcal{T}_{s,\ast,\ast}^{(s)}$ is compactly generated, we will have $\mathcal{T}_{s,\ast,\ast}^{(s)}\otimes \mathcal{T}_{s,\ast,\ast}^{(s)}=\kappa(\mathcal{T}_{s,\ast,\ast}^{(s)}\otimes \mathcal{T}_{s,\ast,\ast}^{(s)})$ and so the channel sum will be continuous on $(\DMC_{\ast,\ast}^{(s)}\times\DMC_{\ast,\ast}^{(s)}, \mathcal{T}_{s,\ast,\ast}^{(s)}\otimes \mathcal{T}_{s,\ast,\ast}^{(s)})$. Note that although $\mathcal{T}_{s,\ast,\ast}^{(s)}$ and $\mathcal{T}_{s,\ast,\ast}^{(s)}$ are compactly generated, their product $\mathcal{T}_{s,\ast,\ast}^{(s)}\otimes \mathcal{T}_{s,\ast,\ast}^{(s)}$ might not be compactly generated.

\section{Discussion and open problems}

The following continuity-related problems remain open:
\begin{itemize}
\item The continuity of the channel parameters $C$ and $P_{e,n,M}$ in the BRM topology $\mathcal{T}_{\ast,\ast}^{(s)}$.
\item The continuity of the channel sum and the channel product on the whole product space $(\DMC_{\ast,\ast}^{(s)}\times \DMC_{\ast,\ast}^{(s)},\mathcal{T}_{s,\ast,\ast}^{(s)}\otimes \mathcal{T}_{s,\ast,\ast}^{(s)})$. As we explained in Section \ref{subsecContOperShan}, it is sufficient to prove that the product topology $\mathcal{T}_{s,\ast,\ast}^{(s)}\otimes \mathcal{T}_{s,\ast,\ast}^{(s)}$ is compactly generated.
\item The continuity of the channel sum and the channel product in the BRM topology.
\end{itemize}

\section*{Acknowledgment}

I would like to thank Emre Telatar for helpful discussions. I am also grateful to Maxim Raginsky for informing me about the work of Blackwell on statistical experiments.

\appendices

\section{Proof of Proposition \ref{propReldXYdXYs}}
\label{appReldXYdXYs}
Fix $n,m\geq 1$ and let $l\in\Delta_{[n]\times[m]}$. Define $\mathcal{G}_1=([n],\mathcal{X},\mathcal{Y},[m],l,W_1)$ and $\mathcal{G}_2=([n],\mathcal{X},\mathcal{Y},[m],l,W_2)$. For every $S\in\mathcal{S}_{[n],\mathcal{X},\mathcal{Y},[m]}$, we have:
\begin{align*}
\hat{\$}&(S,\mathcal{G}_1)\\
&=\frac{1}{n}\sum_{u\in[n]} \hat{\$}(u,S,\mathcal{G}_1)=\frac{1}{n}\sum_{\substack{u\in[n]}}\sum_{i=1}^{n_S}\alpha_S(i)\sum_{y\in\mathcal{Y}}W_1\big(y\big|f_{i,S}(u)\big)l\big(u,g_{i,S}(y)\big)\\
&=\left(\frac{1}{n}\sum_{\substack{u\in[n]}}\sum_{i=1}^{n_S}\alpha_S(i)\sum_{y\in\mathcal{Y}}W_2\big(y\big|f_{i,S}(u)\big)l\big(u,g_{i,S}(y)\big)\right)\\
&\;\;\;+\frac{1}{n}\sum_{\substack{u\in[n]}}\sum_{i=1}^{n_S}\alpha_S(i)\sum_{y\in\mathcal{Y}}\Big(W_1\big(y\big|f_{i,S}(u)\big)-W_2\big(y\big|f_{i,S}(u)\big)\Big)l\big(u,g_{i,S}(y)\big)\\
&\leq \hat{\$}(S,\mathcal{G}_2) + \sum_{i=1}^{n_S}\frac{\alpha_S(i)}{n}\sum_{\substack{u\in[n]}}\sum_{\substack{y\in\mathcal{Y},\\W_1(y|f_{i,S}(u))\geq W_2(y|f_{i,S}(u))}}\Big(W_1\big(y\big|f_{i,S}(u)\big)-W_2\big(y\big|f_{i,S}(u)\big)\Big)l\big(u,g_{i,S}(y)\big)\\
&\stackrel{(a)}{\leq} \hat{\$}(S,\mathcal{G}_2) + \sum_{i=1}^{n_S}\frac{\alpha_S(i)}{n}\sum_{\substack{u\in[n]}}\sum_{\substack{y\in\mathcal{Y},\\W_1(y|f_{i,S}(u))\geq W_2(y|f_{i,S}(u))}}\Big(W_1\big(y\big|f_{i,S}(u)\big)-W_2\big(y\big|f_{i,S}(u)\big)\Big)\\
&= \hat{\$}(S,\mathcal{G}_2) + \sum_{i=1}^{n_S}\frac{\alpha_S(i)}{n}\sum_{\substack{u\in[n]}} \frac{1}{2} \sum_{y\in\mathcal{Y}} \big|W_1\big(y\big|f_{i,S}(u)\big)-W_2\big(y\big|f_{i,S}(u)\big)\big|\\
&\leq \hat{\$}(S,\mathcal{G}_2) + \sum_{i=1}^{n_S}\frac{\alpha_S(i)}{n}\sum_{\substack{u\in[n]}} \max_{x\in\mathcal{X}} \frac{1}{2} \sum_{y\in\mathcal{Y}} |W_1(y|x)-W_2(y|x)|\\
&= \hat{\$}(S,\mathcal{G}_2) + \sum_{i=1}^{n_S}\frac{\alpha_S(i)}{n}\sum_{\substack{u\in[n]}} d_{\mathcal{X},\mathcal{Y}}(W_1,W_2)=\hat{\$}(S,\mathcal{G}_2) + d_{\mathcal{X},\mathcal{Y}}(W_1,W_2)\\
&\leq d_{\mathcal{X},\mathcal{Y}}(W_1,W_2) + \sup_{S'\in\mathcal{S}_{[n],\mathcal{X},\mathcal{Y},[m]}} \hat{\$}(S',\mathcal{G}_2)=d_{\mathcal{X},\mathcal{Y}}(W_1,W_2) + \$_{\opt}(\mathcal{G}_2),
\end{align*}
where (a) follows from the fact that $l(u,g_{i,S}(y))\leq 1$ (because $l\in\Delta_{[n]\times[m]}$). Therefore,
$$\$_{\opt}(\mathcal{G}_1)=\sup_{S\in\mathcal{S}_{[n],\mathcal{X},\mathcal{Y},[m]}} \hat{\$}(S,\mathcal{G}_1)\leq \$_{\opt}(\mathcal{G}_2) + d_{\mathcal{X},\mathcal{Y}}(W_1,W_2),$$
hence
$$\$_{\opt}(\mathcal{G}_1)- \$_{\opt}(\mathcal{G}_2) \leq d_{\mathcal{X},\mathcal{Y}}(W_1,W_2).$$
We can show similarly that $\$_{\opt}(\mathcal{G}_2)- \$_{\opt}(\mathcal{G}_2) \leq d_{\mathcal{X},\mathcal{Y}}(W_1,W_2)$. Therefore, 
$$|\$_{\opt}(l,\hat{W}_1)- \$_{\opt}(l,\hat{W}_2)|=|\$_{\opt}(l,W_1)- \$_{\opt}(l,W_2)|=|\$_{\opt}(\mathcal{G}_1)- \$_{\opt}(\mathcal{G}_2)| \leq d_{\mathcal{X},\mathcal{Y}}(W_1,W_2).$$
We conclude that
$$d_{\mathcal{X},\mathcal{Y}}^{(s)}(\hat{W}_1,\hat{W}_2)=\sup_{\substack{n,m\geq 1,\\l\in\Delta_{[n]\times[m]}}} |\$_{\opt}(l,\hat{W}_1)- \$_{\opt}(l,\hat{W}_2)|\leq d_{\mathcal{X},\mathcal{Y}}(W_1,W_2).$$

\section{Proof of Proposition \ref{propEmbedShanEquiv}}

\label{appEmbedShanEquiv}

Corollary \ref{corEquivChannelSurjInj} implies that $\Proj_2(D_g\circ W\circ D_f)=\Proj_2(D_g\circ W'\circ D_f)$ if and only if $W R_{\mathcal{X}_1,\mathcal{Y}_1}^{(s)}W'$. Therefore, $\Proj_2(D_g\circ W'\circ D_f)$ does not depend on $W'\in\hat{W}$, hence $F$ is well defined. Corollary \ref{corEquivChannelSurjInj} also shows that $\Proj_2(D_g\circ W'\circ D_f)$ does not depend on the particular choice of the surjection $f$ or the injection $g$, hence it is canonical (i.e., it depends only on $\mathcal{X}_1,\mathcal{X}_2,\mathcal{Y}_1$ and $\mathcal{Y}_2$).

On the other hand, the mapping $W\rightarrow D_g\circ W\circ D_f$ is a continuous mapping from $\DMC_{\mathcal{X}_1,\mathcal{Y}_1}$ to $\DMC_{\mathcal{X}_2,\mathcal{Y}_2}$, and $\Proj_2$ is continuous. Therefore, the mapping $W\rightarrow \Proj_2(D_g\circ W\circ D_f)$ is a continuous mapping from $\DMC_{\mathcal{X}_1,\mathcal{Y}_1}$ to $\DMC_{\mathcal{X}_2,\mathcal{Y}_2}^{(s)}$. Now since $\Proj_2(D_g\circ W \circ D_f)$ depends only on the $R_{\mathcal{X}_1,\mathcal{Y}_1}^{(s)}$-equivalence class $\hat{W}$ of $W$, Lemma \ref{lemQuotientFunction} implies that the transcendent mapping of $W\rightarrow \Proj_2(D_g\circ W\circ D_f)$ that is defined on $\DMC_{\mathcal{X}_1,\mathcal{Y}_1}^{(s)}$ is continuous. Therefore, $F$ is a continuous mapping from $(\DMC_{\mathcal{X}_1,\mathcal{Y}_1}^{(s)},\mathcal{T}_{\mathcal{X}_1,\mathcal{Y}_1}^{(s)})$ to $(\DMC_{\mathcal{X}_2,\mathcal{Y}_2}^{(s)},\mathcal{T}_{\mathcal{X}_2,\mathcal{Y}_2}^{(s)})$. Moreover, we can see from Corollary \ref{corEquivChannelSurjInj} that $F$ is an injection.

For every closed subset $B$ of $\DMC_{\mathcal{X}_1,\mathcal{Y}_1}^{(s)}$, $B$ is compact since $\DMC_{\mathcal{X}_1,\mathcal{Y}_1}^{(s)}$ is compact, hence $F(B)$ is compact because $F$ is continuous. This implies that $F(B)$ is closed in $\DMC_{\mathcal{X}_2,\mathcal{Y}_2}^{(s)}$ since $\DMC_{\mathcal{X}_2,\mathcal{Y}_2}^{(s)}$ is Hausdorff (as it is metrizable). Therefore, $F$ is a closed mapping.

Now since $F$ is an injection that is both continuous and closed, $F$ is a homeomorphism between $\DMC_{\mathcal{X}_1,\mathcal{Y}_1}^{(s)}$ and $F\big(\DMC_{\mathcal{X}_1,\mathcal{Y}_1}^{(s)}\big)\subset \DMC_{\mathcal{X}_2,\mathcal{Y}_2}^{(s)}$.

We would like now to show that $F\big(\DMC_{\mathcal{X}_1,\mathcal{Y}_1}^{(s)}\big)$ depends only on $|\mathcal{X}_1|$, $|\mathcal{Y}_1|$, $\mathcal{X}_2$ and $\mathcal{Y}_2$. Let $\mathcal{X}_1'$ and $\mathcal{Y}_1'$ be two finite sets such that $|\mathcal{X}_1|=|\mathcal{X}_1'|$ and $|\mathcal{Y}_1|=|\mathcal{Y}_1'|$. For every $W\in \DMC_{\mathcal{X}_1',\mathcal{Y}_1'}$, let $\overline{W}\in\DMC_{\mathcal{X}_1',\mathcal{Y}_1'}^{(s)}$ be the $R_{\mathcal{X}_1',\mathcal{Y}_1'}^{(s)}$-equivalence class of $W$.

Let $f':\mathcal{X}_1\rightarrow \mathcal{X}_1'$ be a fixed bijection from $\mathcal{X}_1$ to $\mathcal{X}_1'$ and let $f''=f'\circ f$. Also, let $g':\mathcal{Y}_1'\rightarrow\mathcal{Y}_1$ be a fixed bijection from $\mathcal{Y}_1'$ to $\mathcal{Y}_1$ and let $g''=g\circ g'$. Define $F': \DMC_{\mathcal{X}_1',\mathcal{Y}_1'}^{(s)}\rightarrow \DMC_{\mathcal{X}_2,\mathcal{Y}_2}^{(s)}$ as $F'(\overline{W})=\widetilde{D_{g''}\circ W'\circ D_{f''}}=\Proj_2(D_{g''}\circ W'\circ D_{f''}),$ where $W'\in \overline{W}$. As above, $F'$ is well defined, and it is a homeomorphism from $\DMC_{\mathcal{X}_1',\mathcal{Y}_1'}^{(s)}$ to $F'\big(\DMC_{\mathcal{X}_1',\mathcal{Y}_1'}^{(s)}\big)$. We want to show that $F'\big(\DMC_{\mathcal{X}_1',\mathcal{Y}_1'}^{(s)}\big)=F\big(\DMC_{\mathcal{X}_1,\mathcal{Y}_1}^{(s)}\big)$. For every $\overline{W}\in \DMC_{\mathcal{X}_1',\mathcal{Y}_1'}^{(s)}$, let $W'\in\overline{W}$. We have 
\begin{align*}
\textstyle F'(\overline{W})= \textstyle\Proj_2(D_{g''}\circ W'\circ D_{f''})&=\textstyle \Proj_2(D_{g}\circ(D_{g'}\circ W'\circ D_{f'})\circ D_f)\\
&=\textstyle F\left(\widehat{D_{g'}\circ W'\circ D_{f'}}\right)\in F\big(\DMC_{\mathcal{X}_1,\mathcal{Y}_1}^{(s)}\big).
\end{align*}
Since this is true for every $\overline{W}\in \DMC_{\mathcal{X}_1',\mathcal{Y}_1'}^{(s)}$, we deduce that $F'\big(\DMC_{\mathcal{X}_1',\mathcal{Y}_1'}^{(s)}\big)\subset F\big(\DMC_{\mathcal{X}_1,\mathcal{Y}_1}^{(s)}\big)$. By exchanging the roles of $(\mathcal{X}_1,\mathcal{Y}_1)$ and $(\mathcal{X}_1',\mathcal{Y}_1')$ and using the fact that $f=f'^{-1}\circ f''$ and $g=g''\circ g'^{-1}$, we get $F\big(\DMC_{\mathcal{X}_1,\mathcal{Y}_1}^{(s)}\big)\subset F'\big(\DMC_{\mathcal{X}_1',\mathcal{Y}_1'}^{(s)}\big)$. We conclude that $F\big(\DMC_{\mathcal{X}_1,\mathcal{Y}_1}^{(s)}\big)=F'\big(\DMC_{\mathcal{X}_1',\mathcal{Y}_1'}^{(s)}\big)$, which means that $F\big(\DMC_{\mathcal{X}_1,\mathcal{Y}_1}^{(s)}\big)$ depends only on $|\mathcal{X}_1|$, $|\mathcal{Y}_1|$, $\mathcal{X}_2$ and $\mathcal{Y}_2$.

Finally, for every $W'\in\hat{W}$ and every $W''\in F(\hat{W})=\widetilde{D_g\circ W'\circ D_f}$, $W''$ is Shannon-equivalent to $D_g\circ W'\circ D_f$ and $D_g\circ W'\circ D_f$ is Shannon-equivalent to $W'$ (by Lemma \ref{lemEquivChannelSurjInj}), hence $W''$ is Shannon-equivalent to $W'$.

\section{Proof of Lemma \ref{lemDMCXsNorm}}
\label{appDMCXsNorm}

Define $\DMC_{[0],[0]}^{(s)}=\o$, which is strongly closed in $\DMC_{\ast,\ast}^{(s)}$.

Let $A$ and $B$ be two disjoint strongly closed subsets of $\DMC_{\ast,\ast}^{(s)}$. For every $n\geq 0$, let $A_n=A\cap \DMC_{[n],[n]}^{(s)}$ and $B_n=B\cap \DMC_{[n],[n]}^{(s)}$. Since $A$ and $B$ are strongly closed in $\DMC_{\ast,\ast}^{(s)}$, $A_n$ and $B_n$ are closed in $\DMC_{[n],[n]}^{(s)}$. Moreover, $A_n\cap B_n\subset A\cap B=\o$.

Construct the sequences $(U_n)_{n\geq 0},(U_n')_{n\geq 0},(K_n)_{n\geq 0}$ and $(K_n')_{n\geq 0}$ recursively as follows:

$U_0=U_0'=K_0=K_0'=\o\subset\DMC_{[0],[0]}^{(s)}$. Since $A_0=B_0=\o$, we have $A_0\subset U_0\subset K_0$ and $B_0\subset U_0'\subset K_0'$. Moreover, $U_0$ and $U_0'$ are open in $\DMC_{[0],[0]}^{(s)}$, $K_0$ and $K_0'$ are closed in $\DMC_{[0],[0]}^{(s)}$, and $K_0\cap K_0'=\o$.

Now let $n\geq 1$ and assume that we constructed $(U_j)_{0\leq j< n},(U_j')_{0\leq j< n},(K_j)_{0\leq j< n}$ and $(K_j')_{0\leq j< n}$ such that for every $0\leq j< n$, we have $A_j\subset U_j\subset K_j\subset\DMC_{[j],[j]}^{(s)}$, $B_j\subset U_j'\subset K_j'\subset \DMC_{[j],[j]}^{(s)}$, $U_j$ and $U_j'$ are open in $\DMC_{[j],[j]}^{(s)}$, $K_j$ and $K_j'$ are closed in $\DMC_{[j],[j]}^{(s)}$, and $K_j\cap K_j'=\o$. Moreover, assume that $K_j\subset U_{j+1}$ and $K_j'\subset U_{j+1}'$ for every $0\leq j<n-1$.

Let $C_n=A_n\cup K_{n-1}$ and $D_n=B_n\cup K_{n-1}'$. Since $K_{n-1}$ and $K_{n-1}'$ are closed in $\DMC_{[n-1],[n-1]}^{(s)}$ and since $\DMC_{[n-1],[n-1]}^{(s)}$ is closed in $\DMC_{[n],[n]}^{(s)}$, we can see that $K_{n-1}$ and $K_{n-1}'$ are closed in $\DMC_{[n],[n]}^{(s)}$. Therefore, $C_n$ and $D_n$ are closed in $\DMC_{[n],[n]}^{(s)}$. Moreover, we have
\begin{align*}
C_n\cap D_n&=(A_n\cup K_{n-1})\cap(B_n\cup K_{n-1}')\\
&=(A_n\cap B_n)\cup(A_n\cap K_{n-1}')\cup (K_{n-1}\cap B_n)\cup(K_{n-1}\cap K_{n-1}')\\
&\stackrel{(a)}{=}\textstyle \left(A_n\cap K_{n-1}'\cap \DMC_{[n-1],[n-1]}^{(s)}\right)\cup \left(K_{n-1}\cap \DMC_{[n-1],[n-1]}^{(s)}\cap B_n\right)\\
&=(A_{n-1}\cap K_{n-1}')\cup (K_{n-1}\cap B_{n-1})\subset (K_{n-1}\cap K_{n-1}')\cup (K_{n-1}\cap K_{n-1}')=\o,
\end{align*}
where (a) follows from the fact that $A_n\cap B_n=K_{n-1}\cap K_{n-1}'=\o$ and the fact that $K_{n-1}\subset \DMC_{[n-1],[n-1]}^{(s)}$ and $K_{n-1}'\subset \DMC_{[n-1],[n-1]}^{(s)}$.

Since $\DMC_{[n],[n]}^{(s)}$ is normal (because it is metrizable), and since $C_n$ and $D_n$ are closed disjoint subsets of $\DMC_{[n],[n]}^{(s)}$, there exist two sets $U_n,U_n'\subset \DMC_{[n],[n]}^{(s)}$ that are open in $\DMC_{[n],[n]}^{(s)}$ and two sets $K_n,K_n'\subset \DMC_{[n],[n]}^{(s)}$ that are closed in $\DMC_{[n],[n]}^{(s)}$ such that $C_n\subset U_n\subset K_n$, $D_n\subset U_n'\subset K_n'$ and $K_n\cap K_n'=\o$. Clearly, $A_n\subset U_n\subset K_n\subset \DMC_{[n],[n]}^{(s)}$, $B_n\subset U_n'\subset K_n'\subset \DMC_{[n],[n]}^{(s)}$, $K_{n-1}\subset U_n$ and $K_{n-1}'\subset U_n'$. This concludes the recursive construction.

Now define $\displaystyle U=\bigcup_{n\geq 0}U_n=\bigcup_{n\geq 1}U_n$ and $\displaystyle U'=\bigcup_{n\geq 0}U_n'=\bigcup_{n\geq 1}U_n'$. Since $A_n\subset U_n$ for every $n\geq 1$, we have 
\begin{align*}
\textstyle A=A\cap\DMC_{\ast,\ast}^{(s)}=A\cap\left({\displaystyle\bigcup_{n\geq 1}}\DMC_{[n],[n]}^{(s)}\right)={\displaystyle\bigcup_{n\geq 1}}\left(A\cap \DMC_{[n],[n]}^{(s)}\right)={\displaystyle\bigcup_{n\geq 1}} A_n\subset {\displaystyle\bigcup_{n\geq 1}} U_n =U.
\end{align*}
Moreover, for every $n\geq 1$ we have
\begin{align*}
\textstyle U\cap \DMC_{[n],[n]}^{(s)}=\left({\displaystyle\bigcup_{j\geq 1} U_j}\right)\cap \DMC_{[n],[n]}^{(s)}\stackrel{(a)}{=}\left({\displaystyle\bigcup_{j\geq n} U_j}\right)\cap \DMC_{[n],[n]}^{(s)}={\displaystyle\bigcup_{j\geq n} \left(U_j\cap \textstyle\DMC_{[n],[n]}^{(s)}\right)},
\end{align*}
where (a) follows from the fact that $U_j\subset K_j\subset U_{j+1}$ for every $j\geq 0$, which means that the sequence $(U_j)_{j\geq 1}$ is increasing.

For every $j\geq n$, we have $\DMC_{[n],[n]}^{(s)}\subset \DMC_{[j],[j]}^{(s)}$ and $U_j$ is open in $\DMC_{[j],[j]}^{(s)}$, hence $U_j\cap \DMC_{[n],[n]}^{(s)}$ is open in $\DMC_{[n],[n]}^{(s)}$. Therefore, $U\cap \DMC_{[n],[n]}^{(s)}=\displaystyle\bigcup_{j\geq n} \left(U_j\cap \textstyle\DMC_{[n],[n]}^{(s)}\right)$ is open in $\DMC_{[n],[n]}^{(s)}$. Since this is true for every $n\geq 1$, we conclude that $U$ is strongly open in $\DMC_{\ast,\ast}^{(s)}$.

We can show similarly that $B\subset U'$ and that $U'$ is strongly open in $\DMC_{\ast,\ast}^{(s)}$. Finally, we have
\begin{align*}
U\cap U'=\left(\bigcup_{n\geq 1} U_n\right)\cap \left(\bigcup_{n'\geq 1} U_{n'}'\right)=\bigcup_{n\geq 1, n'\geq 1}(U_n\cap U_{n'}')\stackrel{(a)}{=}\bigcup_{n\geq 1}(U_n\cap U_n')
&\subset\bigcup_{n\geq 1}(K_n\cap K_n')=\o,
\end{align*}
where (a) follows from the fact that for every $n\geq 1$ and every $n'\geq 1$, we have $$U_n\cap U_{n'}'\subset U_{\max\{n,n'\}}\cap U_{\max\{n,n'\}}'$$ because $(U_n)_{n\geq 1}$ and $(U_n')_{n\geq 1}$ are increasing. We conclude that $(\DMC_{\ast,\ast}^{(s)},\mathcal{T}_{s,\ast,\ast}^{(s)})$ is normal.

\bibliographystyle{IEEEtran}
\bibliography{bibliofile}

\begin{thebibliography}{10}
\providecommand{\url}[1]{#1}
\csname url@samestyle\endcsname
\providecommand{\newblock}{\relax}
\providecommand{\bibinfo}[2]{#2}
\providecommand{\BIBentrySTDinterwordspacing}{\spaceskip=0pt\relax}
\providecommand{\BIBentryALTinterwordstretchfactor}{4}
\providecommand{\BIBentryALTinterwordspacing}{\spaceskip=\fontdimen2\font plus
\BIBentryALTinterwordstretchfactor\fontdimen3\font minus
  \fontdimen4\font\relax}
\providecommand{\BIBforeignlanguage}[2]{{%
\expandafter\ifx\csname l@#1\endcsname\relax
\typeout{** WARNING: IEEEtran.bst: No hyphenation pattern has been}%
\typeout{** loaded for the language `#1'. Using the pattern for}%
\typeout{** the default language instead.}%
\else
\language=\csname l@#1\endcsname
\fi
#2}}
\providecommand{\BIBdecl}{\relax}
\BIBdecl

\bibitem{ShannonDegrad}
C.~Shannon, ``A note on a partial ordering for communication channels,''
  \emph{Inform. Contr.}, vol.~1, pp. 390--397, 1958.

\bibitem{blackwell1951}
D.~Blackwell, ``Comparison of experiments,'' in \emph{Proceedings of the Second
  Berkeley Symposium on Mathematical Statistics and Probability}.\hskip 1em
  plus 0.5em minus 0.4em\relax University of California Press, 1951, pp.
  93--102.

\bibitem{Sherman}
S.~Sherman, ``On a theorem of hardy, littlewood, polya, and blackwell,''
  \emph{Proceedings of the National Academy of Sciences of the United States of
  America}, vol.~37, no.~12, pp. 826--831, 1951.

\bibitem{Stein}
C.~Stein, ``Notes on a seminar on theoretical statistics. i. comparison of
  experiments,'' \emph{Report, University of Chicago}, 1951.

\bibitem{RajInputDegrad}
\BIBentryALTinterwordspacing
R.~Nasser, ``On the input-degradedness and input-equivalence between
  channels,'' Tech. Rep., 2017. [Online]. Available:
  \url{http://infoscience.epfl.ch/record/225283}
\BIBentrySTDinterwordspacing

\bibitem{RaginskyShannon}
M.~Raginsky, ``Shannon meets blackwell and le cam: Channels, codes, and
  statistical experiments,'' in \emph{2011 IEEE International Symposium on
  Information Theory Proceedings}, July 2011, pp. 1220--1224.

\bibitem{RajDMCTop}
R.~Nasser, ``Topological structures on {DMC} spaces,'' \emph{arXiv:1701.04467},
  Jan 2017.

\bibitem{RajContTop}
------, ``Continuity of channel parameters and operations under various {DMC}
  topologies,'' \emph{arXiv:1701.04466}, Jan 2017.

\bibitem{Engelking}
R.~Engelking, \emph{General topology}, ser. Monografie matematyczne.\hskip 1em
  plus 0.5em minus 0.4em\relax PWN, 1977.

\bibitem{ChannelSumProduct}
C.~Shannon, ``The zero error capacity of a noisy channel,'' \emph{IRE
  Transactions on Information Theory}, vol.~2, no.~3, pp. 8--19, September
  1956.

\bibitem{MiniMax}
D.~Du and P.~Pardalos, \emph{Minimax and Applications}, ser. Nonconvex
  Optimization and Its Applications.\hskip 1em plus 0.5em minus 0.4em\relax
  Springer US, 2013.

\bibitem{CompactlyGenerated}
N.~E. Steenrod, ``A convenient category of topological spaces.'' \emph{Michigan
  Math. J.}, vol.~14, no.~2, pp. 133--152, 05 1967.

\end{thebibliography}
\end{document}